\documentclass[12pt,onecolumn]{IEEEtran}
\usepackage{graphicx}
\usepackage{amssymb, amsmath, graphicx, cite, color, epsfig, subfigure}

\newtheorem{theorem}{Theorem}
\newtheorem{lemma}[theorem]{Lemma}

\newtheorem{definition}{Definition}

\newtheorem{remark}[definition]{Remark}

\begin{document}

\title{Two-way Networks: when Adaptation is Useless}
\author{\IEEEauthorblockN{Zhiyu Cheng, Natasha Devroye\\
University of Illinois at Chicago \\
zcheng3, devroye@uic.edu
\thanks{Portions of this work appeared in Allerton 2011 \cite{zcheng_Allerton},  ISIT 2012 \cite{zcheng_ISIT}, and Allerton 2012 \cite{zcheng_Allerton2012}. The work of Z. Cheng and N. Devroye was partially supported by NSF under award 1053933. The contents of this article are solely the responsibility of the authors and do not necessarily represent the official views of the NSF. 
}
}}

\maketitle

\begin{abstract}
Most wireless communication networks are two-way, where nodes act as both sources and destinations of messages. This allows for ``adaptation'' at or ``interaction'' between the nodes --  a node's channel inputs may be  functions of its message(s) and previously received signals, in contrast to feedback-free one-way channels where inputs are functions of messages only. 
How to best  adapt, or cooperate,  is key to two-way communication, rendering it complex and challenging. 
However, examples exist of channels where adaptation is {\it not} beneficial from a capacity perspective;  it is known that for the point-to-point two-way modulo 2 adder and Gaussian channels,  adaptation does not increase capacity. We ask whether analogous results hold for several multi-user two-way networks. 

We first consider deterministic two-way channel models: the binary modulo-2 addition channel and a generalization of this, and the linear deterministic channel  which models Gaussian channels at high SNR. For these deterministic models we obtain the capacity region for the two-way multiple access/broadcast channel, the two-way Z channel and the two-way interference channel (under certain ``partial'' adaptation constraints in some regimes).  
We permit all nodes to adapt their channel inputs to past outputs (except for portions of the linear high-SNR two-way interference channel where we only permit 2 of the 4 nodes to fully adapt). However, we show that this adaptation is useless from a capacity region perspective. That is, 
the two-way fully or partially adaptive capacity region consists of two parallel ``one-way'' regions operating simultaneously in opposite directions, achieved by strategies where the channel inputs at each use do not adapt to previous inputs. 
We next consider two noisy channel models: first, the  Gaussian two-way MAC/BC, where we show that adaptation can at most increase the sum-rate by $\frac{1}{2}$ bit in each direction. Next, for the two-way interference channel, partial adaptation is shown to be useless when the interference is very strong. In the strong and weak interference regimes, we show that the non-adaptive Han and Kobayashi scheme utilized in parallel in both directions achieves to within a constant gap for the symmetric rate of the fully (for some regimes) or partially (for the remaining regimes) adaptive models.
The central technical contribution is the derivation of new, computable outer bounds which allow for adaptation. Inner bounds follow from known, non-adaptive achievability schemes of the corresponding one-way channel models. 
\end{abstract}

\section{Introduction}

Two-way communication, where users A and B wish to exchange a stream of information, is a natural form of communication of relevance in present and future wireless networks. Applications include two-way high data-rate tele-medicine over wireless broadband links, mobile video conferencing over next generation cellular networks, the synchronization of data among terminals, and communication between a base station and clients. Indeed, much of our current wireless communication is already two-way in nature, but it is not treated as such in practice. Rather, current channel coding schemes orthogonalize the two directions, rendering the two-way channel equivalent to two one-way communication links. 
While this is simple to implement, whether such non-adaptive two-way coding schemes are optimal from a capacity perspective remains an open question. 

What makes two-way communications, in which two (or more) users exchange messages over the same shared channel, challenging are the possibilities that stem from having  nodes be both  sources and destinations of messages. This permits them to adapt their channel inputs to their past received signals. Such two-way adaptation was first considered in the point-to-point two-way channel by Shannon \cite{Shannon:1961}. Shannon's inner and outer bounds \cite{Shannon:1961} are not tight in general, and a general computable\footnote{By computable we mean single-letter expression without the use of unbounded cardinality auxiliary random variables. Multi-letter formulas for the capacity of two-way channels exist, see the expressions involving directed information over code-trees of \cite{Kramer:thesis}.} formula for the capacity region of the point-to-point  two-way channel remains open. 

 \begin{figure*}[h]
\begin{center}
\includegraphics[width=12cm]{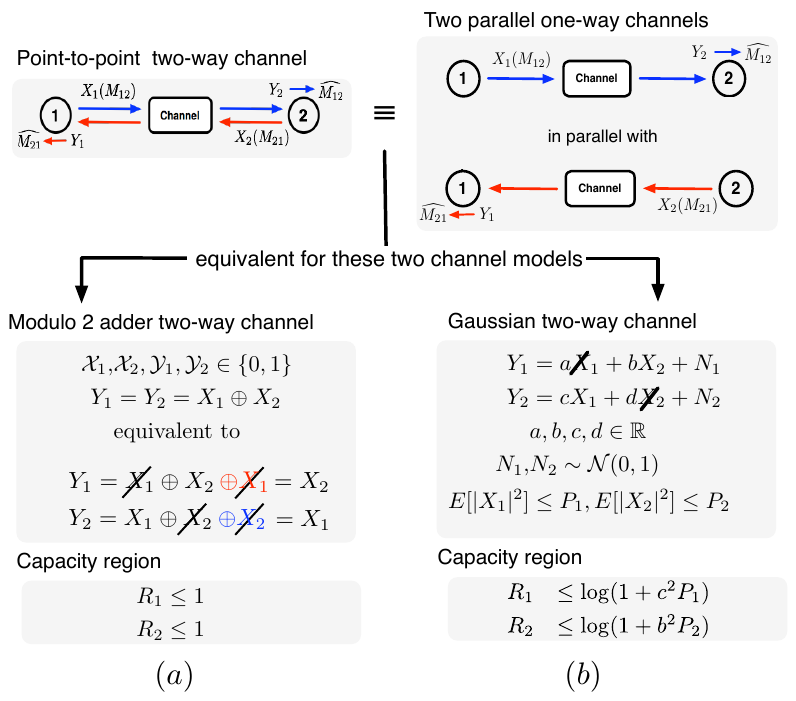}
\caption{Examples of point-to-point two-way channels breaking up into two parallel one-way channels.} 
\label{fig:intro}
\end{center}
\end{figure*}

However, encouragingly, capacity is known for several point-to-point two-way channel models where the interaction between one's own signal and that of the other user may be resolved.  For example, in the two-way modulo 2 binary adder channel where channel outputs  $Y_1 = Y_2 = X_1 \oplus  X_2$ for binary inputs $X_1,X_2$ and  $\oplus$ modulo 2 addition, the capacity region is one bit per user per channel use. Each user is able to ``undo'' the effect of the other as shown in Fig. \ref{fig:intro} (a), something that is not possible in one channel use for the  elusive binary multiplier channel with $Y_1=Y_2 = X_1 X_2$. In the binary modulo 2 adder channel, information independently flows in the $\rightarrow$ and the $\leftarrow$ ``directions''  and nodes need not interact, or adapt their current inputs to past outputs,  to achieve capacity.  In a similar fashion, the capacity of a two-way Gaussian point-to-point channel is equal to two parallel Gaussian channels as shown in Fig. \ref{fig:intro} (b), which may be achieved without the use of adaptation at the nodes \cite{Han:1984}. 

{\bf A note on terminology.} In this work, ``adaptation'' or ``interaction'' is said to take place when the next channel  input of a node is a non-trivial function of that node's past received signals. One may alternatively use the terms ``feedback'' or ``cooperation'' instead of adaptation or interaction. However, we feel that ``adaptation'' and ``interaction'' better highlights the nature of {\it two-way} communications where there is no real notion of feedback (which suggests backwards links which serve to aid communication in the forward direction) as all links may carry information for forward and backwards directions simultaneously. 
``Cooperation'' may also be used and more accurately reflects the fact that nodes may help each other in multi-user two-way channels, but has been 
used in many existing one-way communication scenarios. We feel fresh terminology to emphasize the fact that all nodes may adapt their transmission to each other is useful, as it does not contain any notion of directionality. 



\subsection{Contributions} We seek examples of  multi-user two-way channels rather than point-to-point two-way channels where, even though nodes may adapt current inputs to past outputs, this is {\it not} beneficial from a capacity region perspective. In two-way networks, one may expect adaptation to, in general, be useful and enlarge the capacity region. For example, in multi-user Gaussian channels one may  intuitively expect adaptation to allow for correlation between channel inputs which may translate to coherent gains, or allow for routing messages along different paths.  However, as we will see, there exist multi-user channels for which adaptation is useless.  In particular, we introduce three two-way channel models: 
\begin{enumerate}
\item the {\bf two-way Multiple Access / Broadcast  channel (MAC/BC)} in which there are 4 messages and 3 terminals forming a MAC channel in the $\rightarrow$ direction (2 messages) and a BC channel in the opposite $\leftarrow$ direction (2 messages); 

\item the {\bf two-way Z channel} in which there are 6 messages and  4 terminals forming a Z channel in the $\rightarrow$ direction (3 messages) and another Z channel in the opposite $\leftarrow$ direction (3 messages); 

\item the {\bf two-way interference channel (IC)} with 4 messages and 4 terminals forming an IC in the $\rightarrow$ direction (2 messages) and another IC in the $\leftarrow$ direction (2 messages).
\end{enumerate}

We emphasize that all nodes are permitted to adapt, i.e. channel inputs at node $j$ at time $i$ may be functions of the received signals at node $j$ from times 1 to $i-1$, and that data and ``feedback'' share the same links, i.e. there are no orthogonal feedback links. 
Our central {contributions}  are the derivation of the exact, computable, or approximate (to within a constant gap) capacity region of several two-way networks in which adaptation is useless (or leads to bounded gaps) from a capacity perspective. Typically two-way problems/networks result in multi-letter expressions or auxiliary random variables; our results do not.  

\smallskip
\noindent $\bullet$ We consider {\bf deterministic binary modulo 2 adder channels} for each of the three above channel models. These are the simplest examples of multi-user two-way channels where one might intuitively expect adaptation to be useless. 
For these channel models, and slight generalizations thereof, we obtain outer bounds, and demonstrate that non-adaptive time-sharing schemes between nodes transmitting in the same direction achieves capacity. Nodes transmitting data in opposite directions simultaneously transmit. 

\smallskip

\noindent $\bullet$ We next consider {\bf linear deterministic models} of the three two-way channels above which model Gaussian channels at high SNR  \cite{Avestimehr2009}
 and  again ask whether adaptation may increase the capacity regions beyond that of two parallel one-way multi-user channels in the $\rightarrow$ and $\leftarrow$ directions.  We will show that it does not for the first two channel models by obtaining their capacity regions. For the two-way interference channel, we show that {\it partial adaptation} where only two of the four nodes may adapt,  can ``block'' the two-way information flow and destroy the ability to relay / cooperate,  resulting in a capacity region equal to  two non-adaptive ICs. In addition, in some regimes of the relative link strengths, we obtain the capacity region for the symmetric model with {\it full adaptation} where all four nodes are permitted to adapt.

\smallskip

\noindent $\bullet$ We next consider two noisy Gaussian networks. {First, for the {\bf Gaussian two-way MAC/BC} we demonstrate that adaptation may only increase the sum-rate in each direction by up to $\frac{1}{2}$ bit. } Next, we consider the {\bf symmetric two-way Gaussian IC} where all ``direct'' links are equal and all ``cross-over'' links are equal. We derive new, computable outer bounds for the symmetric sum-rates for this Gaussian channel model and show that: a) adaptation is useless in very strong interference for the partially adaptive model, b) in strong but not very strong interference, non-adaptive schemes perform to within 1 bit per user per direction of the fully adaptive capacity region, and c) the particular non-adaptive Han and Kobayashi scheme of \cite{etkin_tse_wang} employed in each direction, achieves to within a  constant gap (2 bits per user per direction maximally) of fully or partially adaptive outer bounds in all other regimes. 
We provide examples of {\bf non-symmetric Gaussian two-way ICs} where adaptation may provide unbounded gain over non-adaptation, and where perfect output feedback may provide unbounded gain over adaptation. 

The emphasis of this work is on demonstrating when {\it adaptive} schemes are useless, and when, even if adaptation is permitted, it does not significantly increase the capacity region. 

\subsection{Related Work} 
This work builds on: point-to-point two-way channels, one-way multi-user deterministic channels, and one-way multi-user channels with feedback. Little work exists thus far on two-way multi-user channels. 

The capacity region of the general point-to-point discrete memoryless  two-way channel may be written in terms as a limit of multi-letter expressions as in \cite[Section 15]{Shannon:1961}, or \cite[Theorem 4.1]{Kramer:thesis}. Given the complexity in computing this capacity region, it is not entirely satisfying and the capacity region of the two-way channel is  generally considered to be open. The binary multiplier channel (BMC) \cite{Dueck:1979, Schalkwijk:1982, Schalkwijk:1983, Zhang:1986, Hekstra:1989} is a nice example of a deterministic, binary, common output two-way channel where capacity is not exactly known, though its capacity may be expressed in terms of directed information as in \cite[Corollary 4.1]{Kramer:thesis}. However, the capacity regions of particular two-way channels shown in Fig. \ref{fig:intro} are known; in both examples adaptation is useless and the capacity region decomposes into two parallel one-way channels. 
These models were the inspiration for asking whether such examples exist in multi-user two-way networks. 

The first of our three channel models is a two-way MAC/BC channel.   The capacity regions of the linear-deterministic one-way MAC and BC channels were obtained in \cite{bresler_tse}.  An achievable rate region and an outer bound of a similar two-way and adaptive multi-user half-duplex two-way channel is derived in  \cite{Dash:allerton,Dash:asilomar} for Gaussian and discrete memoryless channels (DMC), respectively.  In particular, the achievable rate region derived employs adaptation using Block Markov encoding, and the outer bound contains both auxiliary random variables and messages in its expression and is thus difficult to compute.  
These works differ from our model in that we assume full-duplex operation, have 2 broadcast messages rather than a common one. 
Other than \cite{Dash:allerton,Dash:asilomar}, the two-way MAC/BC has not been considered, and bears most resemblance to a combined MAC channel with feedback and BC channel with feedback (see references in \cite[Ch. 17, Bibliographic Notes]{ElGamalKim:book}, and in particular \cite{willems_FB, elgamal_deg_BC}), though we note that in our two-way model there are no ``free'' feedback links -- any feedback must travel over the same links as the data in the opposite direction, and hence the MAC and BC with feedback results are not directly applicable.

The second channel model we consider is the two-way Z channel, with 6 messages. The one-way Z channel (with 3 messages, rather than the Z Interference channel with 2 messages) was first studied in  \cite{Vishwanath2003Z}, 
 in which a general outer bound, and a matching inner bound for a special class of degraded Z channels are obtained. 
The capacity region of the one-way deterministic Z channel with invertibility constraints similar in flavor to those in \cite{elgamal_det_IC},   is found in \cite{VCadambe2009Z},   which will be of use here.  

The last channel model considered is the two-way linear deterministic IC in which there are 4 messages and 4 terminals forming ICs in the $\rightarrow$ and $\leftarrow$ directions. The capacity region of the one-way modulo 2 adder  IC is known  \cite{ElGamalKim:book} and is a special example of a more general class of deterministic IC for which capacity is known \cite{elgamal_det_IC}, including the one-way linear deterministic IC \cite{bresler_tse}.  
The work here is also related to one-way ICs with perfect output feedback \cite{sahai2009channel, Changho2010}, with rate-limited feedback \cite{Vahid-IC-rate-FB}, with generalized feedback \cite{tuninetti2010outer}, and interfering feedback \cite{sahai2009channel, Suh:ISIT2012}. In all these channel models only two messages are present and the ``feedback'' links, whether perfect, noisy, or interfering still serve only to further rates in the forward direction. The tradeoff between sending new information versus feedback on each of the links is not addressed.  The only other example of such a 4-message two-way interference channel besides our prior work \cite{zcheng_Allerton, zcheng_CISS, zcheng_ISIT, zcheng_Allerton2012} is in Section  VI of \cite{Suh:ISIT2012, Suh2012}, where an example of a linear deterministic scheme is  provided which shows that, at least for one particular asymmetric linear deterministic two-way IC  in weak  interference in the $\rightarrow$ and strong interference in the $\leftarrow$ direction, that adaptation can significantly improve the capacity region over non-interaction. The general capacity region of the linear deterministic two-way IC (with 4 messages)  remains open in general despite the example in \cite{Suh:ISIT2012, Suh2012} and the progress made here. 
One final word on terminology: we will refer to the 4 message two-way IC as the ``two-way IC'' and the 2 message channel of \cite{Suh2012, Suh:ISIT2012} -- considered in all sections but Section VI --  as the ``two-way interference channel with interfering feedback'' to emphasize that the rates are still flowing in one direction only. 
Further comparisons and relationships with ICs with/without feedback \cite{Sahai:2009, Changho2010, Suh2012, Suh:ISIT2012, etkin_tse_wang} will be made in Section \ref{IC} and \ref{ICG}.



\subsection{Outline}
Channel models are first introduced  in Section \ref{model}.  In Section \ref{MACBC} we consider the two-way MAC/BC channel and show that adaptation is useless for three deterministic channel models: 1) the binary modulo 2 adder channel, 2) a generalization of this which we term the ``deterministic, invertible and cardinality constrained'' model, and finally 3) the linear deterministic channel. 
 In Section \ref{Z} we again show that adaptation is useless for the same three deterministic models as in the MAC/BC (but now for the two-way Z channel). 
In Section \ref{IC} we move on to the deterministic IC. For the 1) binary modulo 2 channel we show that adaptation is useless, and show a similar result for its generalization 2) the ``deterministic, invertible and cardinality constrained''  model. For the 3) linear deterministic model we show that adaptation is useless for $2/3<\alpha$ (where $\alpha$ denotes the ratio of cross-over to direct links, as in \cite{etkin_tse_wang}), and show that partial adaptation is useless for the remaining channel conditions. We also obtain the general asymmetric capacity region for the linear deterministic channel model under partial adaptation, which is equal to two parallel ICs in opposite directions.  {In Section \ref{MACBCG} we show that adaptation may only increase the sum-capacity of the two-way Gaussian MAC/BC up to $\frac{1}{2}$ bit per direction. }
Finally, in Section \ref{ICG} we consider the Gaussian two-way IC and show that a non-adaptive scheme achieves within a constant gap (and in one case capacity) of any partially (sometimes fully) adaptive scheme. 
We conclude in Section \ref{conclusion} with some general observations and intuition as to when adaptation is useful, which may be extracted from these examples of two-way multi-user channel models. 

\section{Models, Definitions and Notations}
\label{model}

We consider three multi-user two-way channels, where all nodes act as both transmitters  (encoders) and receivers (decoders), as  shown in Fig. \ref{fig:channels-journal}, and described by:

\smallskip
\noindent $\bullet$ {\it the  two-way MAC/BC channel:} transmitters 1 and 3 send independent messages $M_{12}$ and $M_{32}$ to receiver 2, respectively, forming a MAC in the $\rightarrow$ direction. Transmitter 2 sends independent messages $M_{21}$ and $M_{23}$ to receivers 1 and 3, respectively, forming a BC in the $\leftarrow$ direction.

\smallskip
\noindent $\bullet$ {\it the two-way Z  channel:}  transmitters 1 and 4 send messages $M_{12}$ and $M_{43}$  to receivers 2 and 3 respectively. Transmitters 2 and 3 send messages $(M_{21},M_{23})$ and $(M_{32},M_{34})$ to receivers 1,3 and 2,4 respectively. 

\smallskip
\noindent $\bullet$ {\it  the two-way interference channel:} transmitters 1 and 3 send messages $M_{12}$ and $M_{34}$ to receivers 2 and 4, respectively, forming an IC in the $\rightarrow$ direction. Similarly,  transmitters 2 and 4 send messages $M_{21}$ and $M_{43}$ to receivers 1 and 3 respectively, forming another IC in the $\leftarrow$ direction.

For each of these models, let $M_{jk}$ denote the message from node $j$ to node $k$; all messages are independent and uniformly distributed over $\mathcal{M}_{jk} : =\{1,2,\cdots 2^{nR_{jk}}\}$, where the ranges of $j,k$ depend on the channel model (all subsets of $\{1,2,3,4\}$) and $R_{jk}$ is the rate of transmission from node $j$ to node $k$. For example, in the MAC/BC $R_{12}$ is the rate of message $M_{12}$ but $R_{13}$ and $M_{13}$ do not exist. 

 All channels are assumed to be memoryless and at each channel use, described by the input/output relationships in Fig. \ref{fig:channels-journal}. 
Let $X_j$ and $Y_k$ denote the channel input of node $j$ and output at node $k$ used to describe the model (per channel use). Let $X_{j,i} \; (Y_{j,i})$ denote the channel input (output) at node $j$ at channel use $i$, and $X_j^n : = (X_{1,1}, X_{1,2}, \cdots X_{1,n})$. Let $[x]^+ = \max(0,x)$. 
For the binary modulo 2 adder channels the input and output alphabets are $\{0,1\}$, and $\oplus$ denotes modulo 2 addition.
For the linear deterministic models, the channel inputs and outputs are binary vectors, and all addition will be bit-wise and modulo 2. We furthermore let $S$ denote an $N\times N$ lower shift matrix, where $N = \max(n_{jk})$ over the relevant $j,k$ for each channel model, and $n_{jk}$ defines the number of signal bit levels from transmitter $j$ to receiver $k$.  
We will also consider two other types of channel models: the ``deterministic, invertible and cardinality constrained''  deterministic channel models and the Gaussian two-way MAC/BC and interference channels. We will define those channel models in the appropriate sections. 

A node $j$ is said to employ {\it adaptation} or {\it interaction} if  the channel input at time $i$ is a function of the previously received outputs, 
\begin{equation}X_{j,i}=f_j(M_{j{\cal K}}, Y_j^{i-1}),\label{eq:adapt}\end{equation}
 where $f_j$ ($j\in \{1,2,3,4\}$) are deterministic functions, and $M_{j{\cal K}} := \{M_{jk}| k\in {\cal K}\}$ are the (sets of) messages from node $j$ to all the nodes in ${\cal K} \subset \{1,2,3,4\}$ where ${\cal K}$ depends on the channel model, and may be obtained from Fig. \ref{fig:channels-journal}. 
If a node behaves in a {\it non-adaptive} or {\it restricted} fashion then its inputs are functions of its messages only, i.e. $X_{j,i} = f_j(M_{jk})$. If some nodes adapt while the others do not, we refer to this as {\it partial adaptation}, and will specify which nodes adapt. 
Thus, unless otherwise noted, at each time step $0\leq i \leq n$, for $n$ the blocklength,  encoding functions are specified by the mappings $X_{j,i}(M_{j{\cal K}}, Y_j^{i-1})$. Receiver $k$ uses a decoding function $g_k:\mathcal{Y}_k^n \times \mathcal{M}_{k{\cal I}} \rightarrow\mathcal{\widehat{M}}_{{\cal J}k}$ to obtain estimates of all transmitted messages destined to received $k$, $\mathcal{\widehat{M}}_{{\cal J}k} := \{\widehat{M}_{jk}| j\in {\cal J}, {\cal J}\subset \{1,2,3,4\}\}$ depending on the model,  given knowledge of its own message(s) $M_{k{\cal I}}$ for suitable ${\cal I}\subset\{1,2,3,4\}$, which again depends on model. The capacity region of each channel model is the closure of the set of rate tuples for which there exist encoding and decoding functions (of the appropriate rates) which simultaneously drive the probability that any of the estimated messages is not equal to the true message, to zero as $n\rightarrow \infty$.

 \begin{figure*}
\begin{center}
\includegraphics[width=18cm]{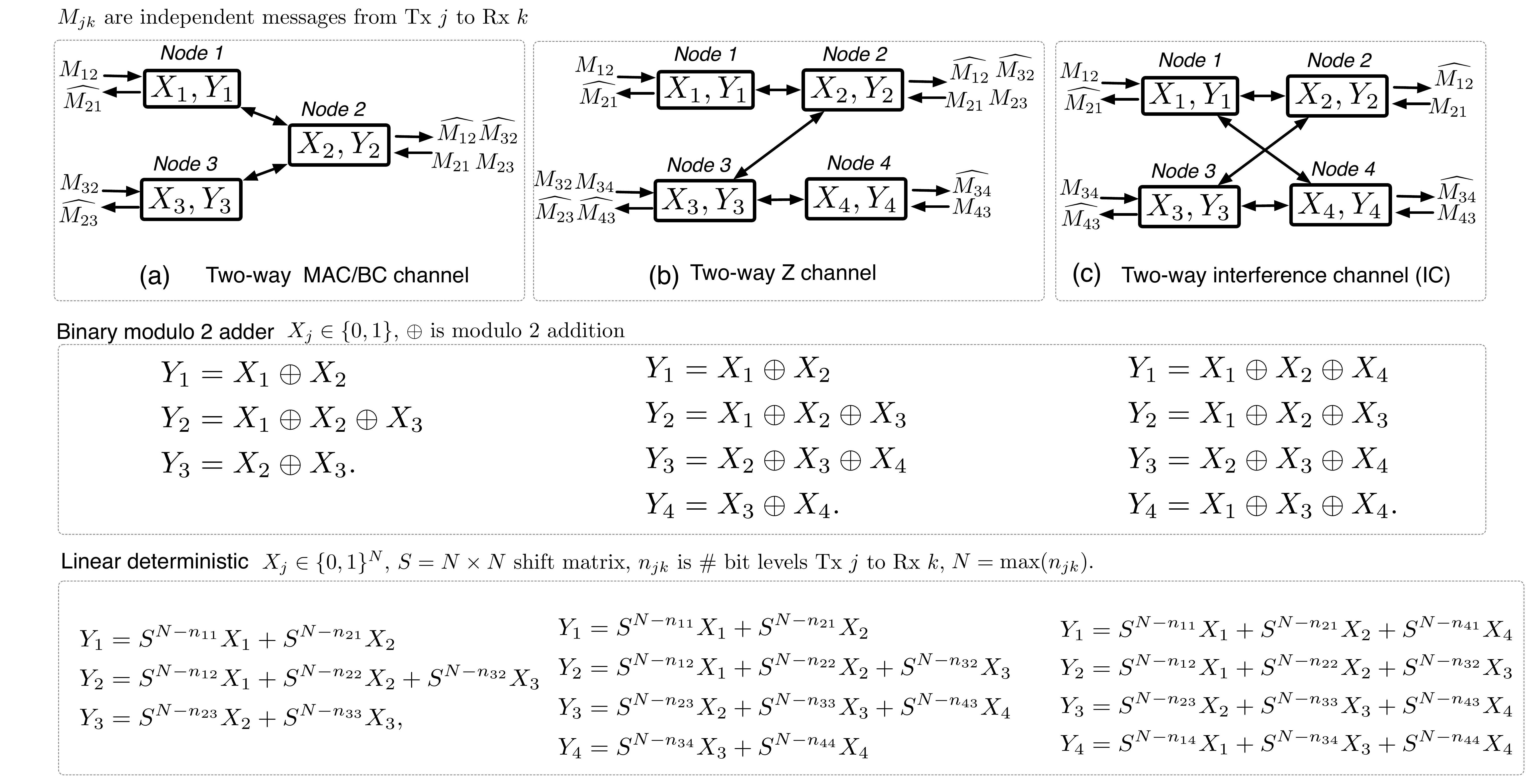}
\caption{Three multi-user two-way channel models and two of the classes under consideration.} 
\label{fig:channels-journal}
\end{center}
\end{figure*}

\section{Two-way MAC/BC}
\label{MACBC}
We first consider the 3 user, full-duplex two-way MAC/BCnetwork as shown in Fig. \ref{fig:channels-journal}(a).
As an introductory example, we first show that adaptation is useless for the modulo 2 adder MAC/BC and a slight generalization thereof, and capacity may be achieved via time-sharing.
Finally, we  consider the linear deterministic two-way MAC/BC and show that, once again, adaptation is useless. 
 
\subsection{An Introductory Example: Modulo 2 Adder MAC/BC}

In the  two-way modulo 2 adder MAC/BC, we emphasize that all three users may employ full adaptation -- i.e. all channel inputs at time $i$ may be a function of previously received channel outputs at that node. There are no additional orthogonal, or free,  ``feedback'' links.  The capacity region may be stated as follows.


\begin{theorem}
\label{thm:3user}
The capacity region of the two-way modulo 2 adder MAC/BC channel is the set of non-negative rate tuples ($R_{12},R_{32},R_{21},R_{23}$) such that
\begin{align}
& R_{12}+R_{32}\leq 1\label{3user1}\\
& R_{23}+R_{21}\leq 1.\label{3user2}
\end{align}
\end{theorem}
\begin{proof}
The outer bound follows from the cut-set bound. The inner bound follows by time-sharing as in Fig. \ref{fig:TS} {\it without adaptation}: $\alpha$ time-shares between channel inputs $X_1$ and $X_3$ for the MAC channel in the $\rightarrow$ direction, while $\beta$ time-shares between the messages $M_{21}$ and $M_{23}$ 
 in the BC in the $\leftarrow$ direction. Both directions ignore the received signals and use i.i.d. Bernoulli(1/2) codebooks.  
\end{proof}



{\begin{remark}
We note that the outer bound may alternatively be derived from either 1) Fano's inequality and first principles, taking into account the ability of the nodes to adapt (we provide one such example for the linear deterministic MAC/BC channel in Theorem \ref{thm:mac/bc} for completeness), or 2)  yet another alternative is to provide Tx 1 and 3 with both $M_{21}, M_{23}$ and perfect channel output feedback $Y_{2,i-1}$  and to provide Tx 2 with channel output feedback $Y_{1,i-1}$ and $Y_{3,i-1}$ at time $i$. Then, one can mimic the outer bound for a class of MACs with FB for 
the $\rightarrow$ direction as derived by Willems \cite{willems_FB} and the physically degraded BC with feedback of \cite{elgamal_deg_BC} for the $\leftarrow$ BC direction (which goes through without a problem for the modulo 2 adder and linear deterministic models). 
Willems' class of channels is one for which (in our notation), at least one of $H(X_1|Y_2,X_3,X_2)$ or $H(X_3|Y_2,X_1,X_2)$ is zero for all input distributions. We note that we provide Tx 1 and Tx 3 with $M_{21}, M_{23}$ in addition to the output feedback in order to be able to construct the inputs $X_{2,i}$, so that $X_2$ would be placed in the conditioning of the bounds of \cite{willems_FB}. 
 Note also that while the capacity of Willems' class of discrete memoryless channels with feedback is expressed in terms of an auxiliary random variable $U$ which is the result of the feedback and its ability to correlate channel inputs. In general, this would result in a larger region than the MAC channel without feedback. However, for our binary modulo 2 channel law, even with conditioning on $X_2$, and the fact that $X_1,X_2,X_3$ may all be correlated, these evaluate to the same region; adaptation is useless. 
\end{remark}
}

%
\begin{remark} The capacity region of Theorem \ref{thm:3user} is the same as that of a modulo 2 adder MAC and a modulo 2 adder BC channel in parallel, which do {\it not} interact. That is, the capacity of a one-way modulo 2 adder MAC is $R_{12}+R_{32}\leq 1$, while that of a one-way modulo 2 adder BC (which is actually just a BC with $Y_1=X_2=Y_3$ is $R_{23}+R_{21}\leq 1$. No adaptation is needed to achieve these regions. In fact, 
we notice that capacity is achieved by time-sharing amongst the data traveling in the same ``direction'' (i.e. between nodes 1 and 3, and between messages $M_{21}$ and $M_{23}$) but not between the two directions themselves. That is, transmission may take place simultaneously between the two directions, as is the case in the point-to-point modulo 2 adder and Gaussian channel models, where no time-sharing is needed between the two transmission directions $\rightarrow$ and $\leftarrow$.
\end{remark}
\subsection{A more general model for deterministic MAC/BC}
Adaptation is useless for the simple modulo 2 adder MAC/BC channel and capacity is achieved using time-sharing in each direction. We ask whether there exists a larger class of channels for which this holds. We answer this positively by considering the ``deterministic, invertible and alphabet restricted'' class of two-way MAC/BCs  with:
\begin{align*}
&Y_1=F_1(X_1,X_2)\\
&Y_2=F_2(X_1,X_2,X_3)\\
&Y_3=F_3(X_2,X_3)
\end{align*}
where $F_m(),m\in \{1,2,3\}$ are deterministic functions which also satisfy
\begin{itemize}
\item {\bf P1:} $|{\cal X}_1| = |{\cal X}_2|= |{\cal X}_3| = |{\cal X}| = |{\cal Y}_1|= |{\cal Y}_2| = |{\cal Y}_3| = |{\cal Y}| = \kappa$ for known $\kappa \in {\mathbb N}^+$. 
\item {\bf P2:} Given $X_1$, $Y_1$ is invertible, i.e. $\exists$ a function $G_1$ s.t. $X_2=G_1(X_1,Y_1)$. Similarly, we assume $\exists G_{21}, G_{23}, G_3$: $X_1 = G_{21}(X_2,X_3,Y_2)$, $X_3 = G_{23}(X_1,X_3,Y_2)$, and $X_2 = G_3(X_3,Y_3)$. These conditions exclude two-way channels such as the binary multiplier channel.
\item {\bf P3:}  $\exists x_3^*$ such that given $X_3 = x_3^*$, $X_1,X_2$ both uniform on their alphabets implies both $Y_1$ and $Y_2$ uniform on their alphabets. Similarly, $\exists x_1^*$ such that given $X_1 = x_1^*$, $X_3,X_2$ both uniform on their alphabets implies $Y_2, Y_3$ uniform on their alphabets. This ensures we can achieve the full $\log(\kappa)$, and is true only for channels with a high degree of symmetry. 
\end{itemize}
Under these conditions (which we only claim are sufficient and not necessarily necessary), the capacity region of the deterministic MAC/BC is given in the following Theorem.
\begin{theorem}
The capacity region of the two-way ``deterministic, invertible and alphabet restricted''  MAC/BC satisfying the conditions {\bf P1}, {\bf P2}, {\bf P3} is the set of non-negative $(R_{12}, R_{32}, R_{21}, R_{43})$ satisfying: 
\begin{align}
&R_{12}+R_{32}\leq \log\kappa \label{general1}\\
&R_{21}+R_{23}\leq \log\kappa, \label{general2}
\end{align}
which may be achieved via time-sharing (in each direction). 
\label{thm:generalMABC}
\end{theorem}
\begin{proof}
The outer bound follows directly from the cut set, or may be directly derived from Fano's inequality or as a result of MAC and BC channels with feedback, as in Remark 1. 
{The restriction on the alphabet sizes condition {\bf P1} prohibits ``coherent gain'' - like phenomena in the outer bound, where correlation between user inputs may be beneficial, as in for example the Gaussian MAC channel with feedback.  
Essentially, this guarantees that the outer bound of Willems for the MAC with feedback \cite{willems_FB} results in the same region as the MAC without feedback. Condition ${\bf P2}$ also guarantees that Willems' outer bound is applicable, though for the outer bound again, the cut-set bound is sufficient. }

Our achievability scheme consists of time-sharing between user 1 and user 3 in the $\rightarrow$ (MAC) direction, while simultaneously time-sharing between sending data to user 1 and 3 in the $\leftarrow$ (BC) direction, as in 
%
%
 Fig. \ref{fig:TS}. There, we see two time-sharing coefficients $0\leq \alpha, \beta\leq 1$, where $\alpha$ time-shares in the $\rightarrow$ direction and $\beta$ in the $\leftarrow$ direction. Let us consider the rates achieved in time slot (1), of duration $\alpha$ (WLOG we have taken $\alpha<\beta$). Node 1 encodes $M_{12}$ into $X_1$ uniformly distributed over the $\log(\kappa)$ input symbols; node 2 encodes $M_{21}$ into $X_2$ uniformly distributed over $\log(\kappa)$ input symbols and node 3 fixes $X_3=x_3^*$ (rate 0). We claim this scheme achieves the rates $R_{12} = R_{21} =  \alpha \log(\kappa)$, $R_{23} = R_{32} = 0$. Consider $R_{12}$: node 2 receives $Y_2=F_2(X_1,X_2,x_3^*)$. Since node 2 knows $X_2$ and knows that $X_3=x_3^*$,  by {\bf P2},  it may construct $X_1 = G_{21}(X_2,x_3^*,Y_2)$ to decode $M_{12}$. By {\bf P3} this may be done at full rate $\alpha \log(\kappa)$. Similar arguments for time slots (2) and (3) demonstrate that the rates in \eqref{general1}-\eqref{general2} are achievable.

\begin{figure}[h]
\begin{center}
\includegraphics[width=3in]{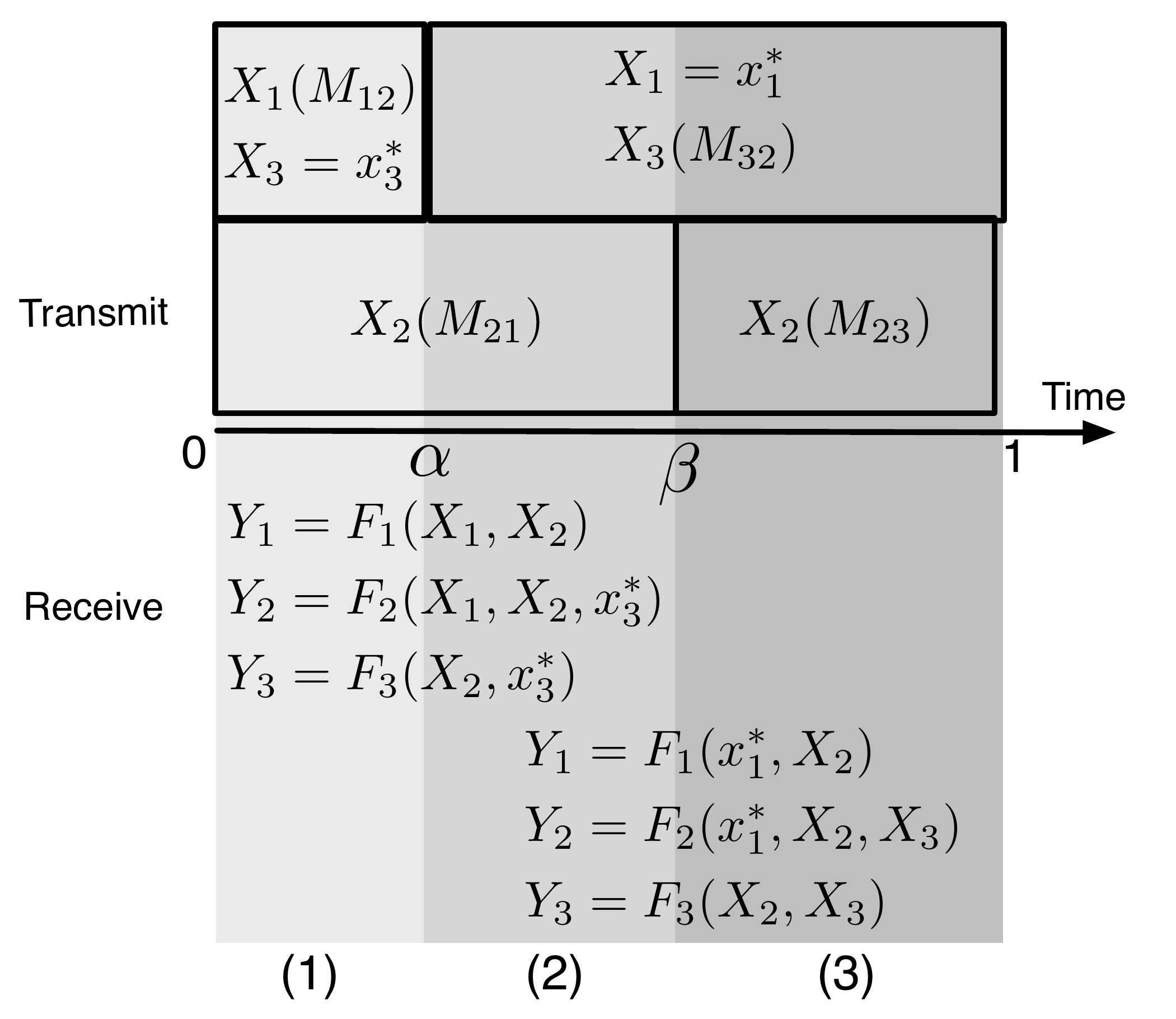}
\caption{Time-sharing based achievability for the proof of Theorem \ref{thm:generalMABC}.} 
\label{fig:TS}
\end{center}
\end{figure}

\end{proof}

One example, besides the binary modulo 2 adder channel, is the channel with input alphabets $\{0,1,\cdots \kappa-1\}$ for some $\kappa$ and $Y_1 = X_1+X_2 \mod \kappa$, $Y_2 = X_1+X_2+X_3 \mod \kappa$, and $Y_3 = X_2+X_3\mod \kappa$. 

\begin{remark}
We comment on restricting the cardinality of the input (which we stress, may not be necessary). This restriction was brought about by simply considering the two-way MAC/BC 
 binary adder channel (not modulo), with inputs ${\cal X}_1 = {\cal X}_3=\{0,1\} $ and outputs $Y_2= X_1+X_3$ with alphabet $\{0,1,2\}$ in the MAC direction. 
In this channel model, it is easy to derive inner and  outer bounds both of the form $R_{12}+R_{32}\leq H(X_1+X_3)$.  In general, one would hope, like Shannon did for the point-to-point two-way channel \cite{Shannon:1961}, to derive multi-user inner and outer bounds of the same {\it form}. However, even if one is able to do so (which may be too much to hope for in general, but may be reasonable for certain classes of deterministic models), we are left with the distributions over which these bounds are taken.  That is, back to our example, to claim that adaptation is useless, the inner bound should be taken over independent input distributions $p(x_1)p(x_3)$, while the outer bound, in general allowing for adaptation,  is taken over joint distributions $p(x_1,x_3)$ (unless one restricts the set of input distributions perhaps via dependence-balance-bound-like techniques \cite{Hekstra:1989}, an open problem). Inner and outer bounds taken over these different sets of distributions do {\it not} match for the binary adder channel with ternary output. As such, we restricted the channels to those for which a form of cooperation (or adaptation) between users cannot possibly help -- which is the case for the modulo adder channels, and as we will see, the similar, in terms of properties, linear deterministic channels. 


%
\end{remark}

\subsection{Linear Deterministic MAC/BC}
The two-way linear deterministic MAC/BC channel is defined by the input/output equations as in Fig. \ref{fig:channels-journal}(a).
We recall that all nodes are permitted to adapt, so that at channel use $i$, 
$ X_{1,i}=f_1(M_{12},Y_1^{i-1})$, $X_{2,i}=f_2(M_{21},M_{23}, Y_2^{i-1})$, 
and $ X_{3,i}=f_3(M_{32},Y_3^{i-1})$.
The capacity region may be stated as follows:

\begin{theorem}
\label{thm:mac/bc}
The capacity region of the two-way linear deterministic MAC/BC is the set of non-negative rate tuples $(R_{12},R_{32},R_{21},R_{23})$ such that
\begin{equation}  \mbox{MAC } \rightarrow \left\{ \begin{array}{l}
 R_{12}\leq n_{12}, \;\; R_{32}\leq n_{32},\\
 R_{12}+R_{32}\leq \max (n_{12},n_{32})
 \end{array}\right. \label{1}
\end{equation}
\begin{equation}  \mbox{BC } \leftarrow \left\{ \begin{array}{l}
 R_{21}\leq n_{21}, \;\;  R_{23}\leq n_{23}\\
R_{21}+R_{23}\leq \max (n_{21},n_{23}).
 \end{array}\right. \label{2}
\end{equation}
\end{theorem}

\begin{proof}
Achievability may be argued via \cite{Avestimehr2009} by mimicking a one-way MAC and one-way BC channel in opposite directions and noting that each user may subtract off its own transmitted signal from its received signal. 
The outer bounds may be obtained by the cut-set, or via an alternative direct proof. We include an example of this alternative (to the cut-set) outer bound proof below, out of interest and to demonstrate how adaptation may be taken into account.
\begin{align*}
& n(R_{12}-\epsilon) \leq I(M_{12};Y_2^n|M_{21},M_{23},M_{32})\\
& \overset{(a)}{\leq} \sum_{i=1}^n [H(Y_{2,i}|Y_2^{i-1},M_{21},M_{23},M_{32},X_2^i)]\\
& \overset{(b)}{=} \sum_{i=1}^n [H(Y_{2,i}|Y_2^{i-1},M_{21},M_{23},M_{32},X_2^i,X_3^i)]\\
& \leq \sum_{i=1}^n [H(S^{N-n_{12}}X_{1,i})]\leq n \, (n_{12}),
\end{align*}
where (a) follows since given $(Y_2^{i-2},M_{21},M_{23})$, we may construct $X_2^i$, which cancels out the ``self-interference'' term $X_{2,i}$ in $Y_{2,i}$. We note that the self-interference term can be always cancelled out in this way in the converse of additive models. Step (b) follows from the fact that given $M_{32}, X_2^i$, we may construct $X_3^i$. The other single rate bounds follow similarly. 
\begin{align*}
&n(R_{12}+R_{32}-\epsilon) \leq I(M_{12},M_{32};Y_2^n|M_{21},M_{23})\\
& \leq \sum_{i=1}^n [H(Y_{2,i}|Y_2^{i-1},M_{21},M_{23},X_2^i)]\\
& \leq \sum_{i=1}^n [H(S^{N-n_{12}}X_{1,i}+S^{N-n_{32}}X_{3,i})]\leq n (\max (n_{12},n_{32})).
\end{align*}
We may analogously obtain the other sum-rate bound.
\end{proof}

\begin{remark}
Without adaptation, the channel would correspond to a MAC channel simultaneously transmitting with a BC channel with restricted nodes. 
This coincides with our outer bound with adaptation, which may furthermore be achieved in one channel use: adaptation is useless, and  the capacity region is a four dimensional region that is equivalent to the capacity region of the linear deterministic MAC and the linear deterministic BC in opposite directions.
\end{remark}

\section{Two-way Z Channel}
\label{Z}
We now consider the 4 user, full-duplex network as shown in Fig. \ref{fig:channels-journal}(b). The 6 message network, resembles a cascade of three two-way channels, in the shape of a Z (in each direction). Again, we first introduce the modulo 2 adder model and its generalization, and then the linear deterministic model, showing that adaptation is useless through the re-derivation of adaptive outer bounds. 
 \subsection{An Introductory Example: Modulo 2 Adder Two-way Z Channel}
 
The two-way modulo 2 adder Z channel is discrete and memoryless, and all four users may employ full adaptation.
The capacity region of this channel is stated as follows: 
\begin{theorem}
\label{mod2thm:z}
The capacity region of the two-way modulo 2 adder Z  channel is the set of non-negative rate tuples ($R_{12},R_{21},R_{23},R_{32},R_{34},R_{43}$) such that
\begin{align}
& R_{12}+R_{32}+R_{34}\leq 1\label{z1}\\
& R_{21}+R_{23}+R_{43}\leq 1\label{z2}
\end{align}
\end{theorem}
The proof is found in Appendix \ref{mod2Z} and is not a direct consequence of the cut-set outer bound. 
\begin{remark}
\label{remark:Ztype}
We note that the proof of the sum-rate outer bound of the $Z$ channel in Theorems \ref{mod2thm:z}, \ref{thm:z}, and the sum-rate bounds of the two-way IC in Theorems \ref{mod2thm:intf}, \ref{thm:ICG}, \ref{thm:intf}, \ref{thm:outer-full} all follow the same general idea of giving an asymmetric genie to one receiver, as initially done in \cite{elgamal_det_IC} for the one-way IC, and quite similar to the $Z$ channel outer bound in \cite{VCadambe2009Z},  and in particular \cite{sahai2009channel, Changho2010} for the one-way IC with feedback (IC with FB). That is, in the $\rightarrow$ direction, we provide one of the receivers with the message of the non-desired message in the $\rightarrow$ direction (as in the IC with FB) as well as all messages of the $\leftarrow$ direction (particular to the two-way channels, as no $\leftarrow$ messages in one-way channels), and the desired signal received at the other receiver of the $\rightarrow$ direction (similar to the genies given in the IC with FB). The additional messages and receiver output (relative to one-way models) are needed to create various inputs, as may be done with less side-information in one-way models. \end{remark}

\begin{remark} We again notice that since time-sharing achieves the above region, adaptation does not enlarge the capacity region.  We again see that the messages in the $\rightarrow$ and the $\leftarrow$ directions may be simultaneously communicated, but that the messages within one direction must be time-shared. 
\end{remark}
%

\subsection{A More General Model for Two-way Z Channel}
Similar to the more general ``deterministic, invertible and restricted'' class of two-way MAC/BC channels where it was shown that non-adaptive time-sharing achieves capacity, we may extend the two-way Z modulo 2 adder model to a more general class of two-way Z channels.  The converse follows along similar lines as for the modulo 2 adder channel. In terms of achieving the outer bounds  $R_{12}+R_{32}+R_{34}\leq \log\kappa$ and $R_{21}+R_{23}+R_{43}\leq \log \kappa$, one sufficient condition involves restricting the input and output alphabet sizes to be equal (eliminating some of the potential benefits of adaptation via user cooperation), as well as several symmetry constraints akin to extensions of {\bf P2} and {\bf P3}. Again, one example of such a channel model is the modulo $\kappa$ channel. We omit the full statement as it follows in a straightforward and analogous fashion to Theorems \ref{thm:generalMABC} and \ref{mod2thm:z}.

\subsection{Linear Deterministic Two-way Z Channel}

The linear deterministic two-way Z channel is defined by the input / output equations in Fig. \ref{fig:channels-journal}(b).
The capacity region is again that of two parallel Z channels in opposite directions; adaptation is useless.
\begin{theorem} 
\label{thm:z}
The capacity region of the two-way linear deterministic Z channel is the set of all rate-tuples ($R_{12},R_{21},R_{23},R_{32},R_{34},R_{43}$) which satisfy the following:

\begin{equation*}  \mbox{Z } \rightarrow \left\{ \begin{array}{l}
 R_{12}\leq n_{12},  \;\; R_{32}\leq n_{32}, \;\; R_{34}\leq n_{34}\\
 R_{12}+R_{32}\leq \max (n_{12},n_{32})\\
   R_{32}+R_{34}\leq \max (n_{32},n_{34})\\
   R_{12}+R_{32}+R_{34}\leq \max (n_{12},n_{32})+[n_{34}-n_{32}]^+
 \end{array}\right.
\end{equation*}
\begin{equation*}  \mbox{Z } \leftarrow \left\{ \begin{array}{l}
  R_{43}\leq n_{43}, \;\; R_{23}\leq n_{23}, \;\; R_{21}\leq n_{21}\\
R_{43}+R_{23}\leq \max (n_{43},n_{23})\\
 R_{23}+R_{21}\leq \max (n_{23},n_{21})\\
   R_{43}+R_{23}+R_{21}\leq \max (n_{43},n_{23})+[n_{21}-n_{23}]^+. \end{array}\right.
\end{equation*}


\end{theorem}
%

\begin{proof}
We first note that the capacity of a class of deterministic Z channels is shown in \cite[Th. 3.1]{VCadambe2009Z}. To show achievability of the above, we use the achievability scheme of \cite[Th. 3.1]{VCadambe2009Z} in each $\rightarrow$ and $\leftarrow$ direction with non-adaptive nodes.
Due to the additive nature of the channel, each receiver may cancel or subtract out its own ``self-interference'' term $S^{N-n_{jj}}X_j$ from its received signal. By making the appropriate correspondences,  the above is achievable and equivalent to two one-way Z channels. 

For the converse,  note that all but the triple-rate bounds may be obtained by the cut-set bound, or independently by giving the appropriate side-information or genie to the receivers (as illustrated in previous models). 
The non-cut-set triple rate bound may be obtained as follows:
\begin{align*}
& n(R_{12}+R_{32}+R_{34}-\epsilon) \leq I(M_{12};Y_2^n|M_{21},M_{23},M_{43}) +I(M_{32},M_{34};Y_2^n, Y_4^n|M_{43},M_{12},M_{21},M_{23}) \\
&\leq H(Y_2^n|M_{21},M_{23},M_{43})+H(Y_4^n|M_{43},M_{12},M_{21},M_{23},Y_2^n)\\
&\leq \sum_{i=1}^n [H(Y_{2,i}|Y_2^{i-1},M_{21},M_{23},X_2^i)+H(Y_{4,i}|M_{12},M_{21},M_{23},M_{43},Y_4^{i-1},X_4^i,Y_2^n, X_2^n)]
\end{align*}
\begin{align*}
& \overset{(a)}{\leq} \sum_{i=1}^n [H(S^{N-n_{12}}X_{1,i}+S^{N-n_{32}}X_{3,i})\\
& \ \ +H(S^{N-n_{34}}X_{3,i}|M_{12},M_{21},M_{23},M_{43},Y_4^{i-1},X_4^i,S^{N-n_{12}}X_{1,i}+S^{N-n_{22}}X_{2,i}+S^{N-n_{32}}X_{3,i},X_2^n,X_1^n)]\\
& \leq \sum_{i=1}^n [H(S^{N-n_{12}}X_{1,i}+S^{N-n_{32}}X_{3,i}) +H(S^{N-n_{34}}X_{3,i}|S^{N-n_{32}}X_{3,i})]\\
&\leq n(\max (n_{12},n_{32})+[n_{34}-n_{32}]^+).
\end{align*}
In (a), $X_1^n$ in the second entropy term follows since given, $M_{12}$ and $X_2^n$, we may construct $X_1^n$.
\end{proof}

\begin{remark}
Again, we are always able to achieve the desired rates in Theorem \ref{thm:z} in only one channel use, therefore adaptation is useless.   The capacity region of this channel, a 6 dimensional region, is exactly equivalent to the capacity region of the two one-way linear deterministic Z channels. 
 \end{remark}
 
\section{Deterministic Two-way Interference Channels}
\label{IC}
The last deterministic multi-user two-way network we consider is a 4 user, 4 message, full-duplex network as shown in Fig. \ref{fig:channels-journal}(c).
 This channel model merges elements of two-way, feedback,  and  interference, and forms two parallel interference channels in the $\rightarrow$ and $\leftarrow$ directions. Again, we first introduce the modulo 2 adder model of this channel and show that adaptation is useless, generalizing this to a slightly larger class of symmetric channels. This generalization is not as straightforward as for the  MAC/BC and Z channels, and hence is discussed in somewhat more depth. 
 Finally, for the symmetric linear deterministic two-way interference channel, we show that full adaptation is useless when the interference is very strong, strong, and in some of the weak regimes, while in all other regimes we show that partial adaptation is useless (i.e. if only 2 of the nodes adapt, might as well have none of the nodes adapt). 
\subsection{An Introductory Example: Modulo 2 Adder Two-way IC}
We are again motivated by the two-way, modulo 2 adder IC, perhaps the simplest example of a two-way IC channel in which adaptation is useless, and capacity is achieved through time-sharing.
\begin{theorem}
\label{mod2thm:intf}
The capacity region of the two-way modulo 2 adder interference channel is the set of non-negative rate tules ($R_{12},R_{21},R_{34},R_{43}$) such that
\begin{align}
& R_{12}+R_{34}\leq 1\label{intf1}\\
& R_{21}+R_{43}\leq 1.\label{intf2}
\end{align}
\end{theorem}
\begin{proof}
We may achieve this region using two time-sharing random variables; one between nodes 1 and 3, and a second between nodes 2 and 4. The converse follows by the cut-set bound, or may alternatively be derived as done in the next subsection for a more general class of channels.  
%
%
\end{proof}

%

\subsection{Comments on a more general class of two-way deterministic ICs}
We ask whether the above two-way modulo 2 IC results may be extended to a more general class of deterministic ICs in which adaptation is useless and capacity is achieved through time-sharing. In both the MAC/BC and Z channel models we were able to accomplish this by imposing certain cardinality, invertibility and symmetry constraints. One example of a channel in this class is the modulo-$\kappa$ (for some $\kappa$) channel. We now extend results to the two-way IC, but note that we must make two additional restrictions: 1) we do not consider ``self-interference'' (which we did in the previous two models), and 2) we impose symmetry of the outputs (common output in each direction). 
Both of these conditions are sufficient for obtaining sum-rate outer bounds equal to $\log\kappa$ in each direction (where $\kappa$ is the input/output alphabet size); whether they are necessary remains open. 

Consider a class of deterministic two-way interference channels without self-interference,
described by:
\begin{align*}
&Y_1=F_\rightarrow(X_2,X_4) = Y_3  \mbox{ (there is no self-interference, symmetric channel)} \\
&Y_2=F_\leftarrow(X_1,X_3) = Y_4  \mbox{ (there is no self-interference, symmetric channel)}
\end{align*}
where $F_\rightarrow, F_\leftarrow$ are deterministic functions. Further restrict the class of channels to those with:
\begin{itemize}
\item {\bf P1IC:} $|{\cal X}_1| = |{\cal X}_2|= |{\cal X}_3| = |{\cal X}_4| = |{\cal Y}_1|= |{\cal Y}_2| = |{\cal Y}_3| =  |{\cal Y}_4| = \kappa$ for known $\kappa \in {\mathbb N}^+$. 
\item {\bf P2IC:} ``Invertibility'' constraints reminiscent of Costa and El Gamal \cite{elgamal_det_IC}.  In the notation of \cite{elgamal_det_IC}, we assume $f_1 = f_2 = F_\rightarrow$ (and similarly, in the reverse direction we have $f_1=f_2=F_\leftarrow$), and that $g_1 = g_2$ are the identity functions, i.e. $g_1(X_1) = X_1$ and $g_2(X_3) = X_3$ (and similarly for the reverse direction). 
Then we require that, given $X_1$, $Y_2$ is invertible, i.e. $\exists$ a function $G_{2}$ s.t. $X_3=G_2(X_1,Y_2)$. Similarly, we assume $\exists G_{1},G_3, G_4$: $X_4 = G_1(X_2,Y_1)$, $X_2 = G_{3}(X_4,Y_3)$, and $X_1 = G_{4}(X_3,Y_4)$. 
\item {\bf P3IC:}  to ensure that we may attain the outer bound through time-sharing, we impose that $F_\rightarrow$ is a function such that $\exists x_3^*$ such that $X_1$ and $Y_2 = Y_4 =  F_\rightarrow(X_1,X_3=x_3^*)$ are in 1-to-1 correspondence, and $\exists  x_1^*$ such that $X_3$ and $Y_2 = Y_4 =  F_\rightarrow(X_1=x_1^*,X_3)$ are in 1-to-1 correspondence,  (and similarly for $F_\leftarrow$). 
\end{itemize}

For this class of channels, the capacity is given by the following:
\begin{theorem}
The capacity region of the two-way ``deterministic, invertible and alphabet restricted''  IC satisfying the conditions {\bf P1IC}, {\bf P2IC}, {\bf P3IC} is the set of non-negative rates $(R_{12}, R_{34}, R_{21}, R_{43})$ satisfying: 
\begin{align}
&R_{12}+R_{34}\leq \log\kappa \label{general1IC}\\
&R_{21}+R_{43}\leq \log\kappa, \label{general2IC}
\end{align}
which may be achieved via time-sharing (in each direction). 
\label{thm:ICG}
\end{theorem}
\begin{proof}
Let us consider only the $\rightarrow$ direction for now. Under the above restrictions, the capacity region of the class of deterministic (one-way) interference channels in \cite{elgamal_det_IC} may be simplified to 
\begin{equation} R_{12}+R_{34} \leq \log\kappa  \label{eq:sum_in} \end{equation}
which may be achieved by time-sharing between the inputs $X_1$ uniform over the $\kappa$ input symbols, while $X_3=x_3^*$ and vice versa. 
That the rates \eqref{eq:sum_in} are achievable may alternatively be directly verified. 

We find the matching outer bound:
\begin{align*}
&n(R_{12}+R_{34}-\epsilon)\\
&\leq I(M_{12};Y_2^n|M_{21},M_{43})+I(M_{34};Y_4^n,Y_2^n|M_{12},M_{21},M_{43})\\
&= I(M_{12};Y_2^n|M_{21},M_{43})+I(M_{34};Y_2^n|M_{21},M_{12},M_{43})+I(M_{34};Y_4^n|M_{21},M_{12},M_{43},Y_2^n)\\
&\overset{(a)}{\leq} \sum_{i=1}^n [H(Y_{2,i}|Y_2^{i-1},M_{21},M_{43})-H(Y_{2,i}|Y_2^{i-1},M_{12},M_{21},M_{43})+H(Y_{2,i}|Y_2^{i-1},M_{12},M_{21},M_{43})]\\
&\leq \sum_{i=1}^n [H(Y_{2.i})]\\
&= n\max_{p(x_1,x_3)} H(Y_2) \leq n \log \kappa,
\end{align*}
where in (a) we dropped two negative entropy terms, and were able to replace $Y_{4,i}$ by $Y_{2,i}$, allowing us to cancel the 2nd and 3rd terms. This is the central reason why we have restricted $Y_2=Y_4$ and $Y_1=Y_3$, whether one may somehow cancel these terms when this is not the case is open.  Restricting the alphabet size as in {\bf P1IC} yields the final inequality. 

\end{proof}

\begin{remark}
We have proposed a slightly more general model for deterministic two-way ICs in which adaptation is useless. However, it should be pointed out that our conditions are sufficient but by no means necessary.  For instance, consider a binary multiplier two-way interference channel described by $Y_1=Y_3=X_2X_4$ and $Y_2=Y_4=X_1X_3$, with all inputs and outputs binary. It is not difficult to show that adaptation is useless for this model and the capacity of this channel is equivalent to the capacity of two one-way binary multiplier interference channels in parallel,  the same capacity region as in Theorem \ref{mod2thm:intf}. In addition, we will show in Section \ref{ICG} that adaptation is also useless for the Gaussian two-way interference channel with partial adaptation when the two-way interference is very strong; this channel is not in the class of channels considered above either. 
\end{remark}
%


\subsection{Linear Deterministic Two-way IC}
The two-way linear deterministic interference channel is defined by the input / output equations in Fig. \ref{fig:channels-journal}(c). 
In this section we will be considering the general linear deterministic IC, as well as the ``symmetric'' linear deterministic IC for which $p:=n_{12}=n_{21}=n_{34}=n_{43}$,  $q:=n_{14}=n_{41}=n_{23}=n_{32}$, and $\alpha: = q/p$. This will allow us to compare the symmetric, normalized sum capacity of various one and two-way interference channels, defined as $C_{sym}(\alpha) : = \frac{R_{12}+R_{34}}{2}$.

Recall our definition of partial adaptation (nodes 1 and 3 are fixed or ``restricted'') of Section \ref{model}:
\begin{align}
& X_{1,i}=f_1(M_{12}), \;\; X_{2,i}=f_2(M_{21},Y_2^{i-1}) \label{eq:partial1}\\
& X_{3,i}=f_3(M_{34}), \;\; X_{4,i}=f_4(M_{43},Y_4^{i-1}) \label{eq:partial2}
\end{align}
We first prove a  Lemma regarding partial adaptation, which is key in showing that partial adaptation is useless, and that the inability of {\it certain} nodes to adapt essentially ``blocks'' the ability of adaptation to help at all. 
\begin{lemma}
Under partial adaptation conditions \eqref{eq:partial1} -- \eqref{eq:partial2}, for some deterministic functions $f_5$ and $f_6$, 
\begin{align}
X_{2,i}& = f_5(M_{12}, M_{21}, M_{34}) \perp M_{43}, \;\; \forall i \label{eq:X2}\\
X_{4,i}& = f_6(M_{43}, M_{34}, M_{12}) \perp M_{21}, \;\; \forall i \label{eq:X4}
\end{align}
where $\perp$ denotes independence.
\label{lemma:partial}
\end{lemma}
\begin{proof}
Note that $X_{2,i} = f_2(M_{21}, Y_2^{i-1})$ and $Y_2^{i-1} = S^{N - n_{12}}X_{1}^{i-1} + S^{N-n_{22}} X_2^{i-1}+S^{N-n_{32}}X_3^{i-1}$. Since $X_1^{i-1}$ and $X_3^{i-1}$ are functions only of $M_{12}$ and $M_{34}$ respectively, we may conclude that  there exists a function $f^*$ such that $X_{2,i} = f^*(M_{21}, M_{12}, M_{34}, X_2^{i-1})$. Iterating this argument, and noting that $X_{2,1}$ is only a function of $M_{21}$, we obtain the theorem. The result for $X_{4,i}$ follows by a similar argument.
\end{proof}

\medskip
\begin{theorem}
\label{thm:intf}
The capacity region of the two-way linear deterministic interference channel under partial adaptation constraints is the set of ($R_{12},R_{21},R_{34},R_{43}$) which satisfy the following:

\begin{align}  
 &R_{12}\leq n_{12},   \;\; R_{34}\leq n_{34}\tag{IC$\rightarrow$ a}\\
 &R_{12}+R_{34}\leq \max (n_{12},n_{32})+[n_{34}-n_{32}]^+\tag{IC$\rightarrow$ b}\\
  &R_{12}+R_{34}\leq \max (n_{34},n_{14})+[n_{12}-n_{14}]^+\tag{IC$\rightarrow$ c}\\
   &R_{12}+R_{34}\leq \max ([n_{12}-n_{14}]^+,n_{32})+\max ([n_{34}-n_{32}]^+,n_{14})\tag{IC$\rightarrow$ d}\\
   &2R_{12}+R_{34}\leq \max (n_{12},n_{32})+[n_{12}-n_{14}]^++\max ([n_{34}-n_{32}]^+,n_{14})\tag{IC$\rightarrow$ e}\\
   &R_{12}+2R_{34}\leq \max (n_{34},n_{14})+[n_{34}-n_{32}]^++\max ([n_{12}-n_{14}]^+,n_{32})\tag{IC$\rightarrow$ f}
\end{align}
\begin{align}  
 & R_{21}\leq n_{21},  \;\; R_{43}\leq n_{43}\tag{IC$\leftarrow$ a}\\
 &R_{21}+R_{43}\leq \max (n_{21},n_{41})+[n_{43}-n_{41}]^+\tag{IC$\leftarrow$ b}\\
  &R_{21}+R_{43}\leq \max (n_{43},n_{23})+[n_{21}-n_{23}]^+\tag{IC$\leftarrow$ c}\\
   &R_{21}+R_{43}\leq \max ([n_{21}-n_{23}]^+,n_{41})+\max ([n_{43}-n_{41}]^+,n_{23})\tag{IC$\leftarrow$ d}\\
   &2R_{21}+R_{43}\leq \max (n_{21},n_{41})+[n_{21}-n_{23}]^++\max ([n_{43}-n_{41}]^+,n_{23})\tag{IC$\leftarrow$ e}\\
   &R_{21}+2R_{43}\leq \max (n_{43},n_{23})+[n_{43}-n_{41}]^++\max ([n_{21}-n_{23}]^+,n_{41}). \tag{IC$\leftarrow$ f}
\end{align}
\end{theorem}

\begin{proof}
For achievability, note that self-interference may be removed at each receiver due to this channel model's linearity, in which case the physical channel model reduces to two one-way IC in opposite directions. We may thus apply the well-known Han-Kobayashi scheme \cite{H+K} in each direction, ignoring the ability of the nodes to adapt, achieving the expression in (IC$\rightarrow$) and (IC$\leftarrow$).

Now we prove the converse. Single-rates follow as in \eqref{1}, and using Lemma \ref{lemma:partial} (where we use partial adaptation). For the sum-rate (IC$\rightarrow$ b): 
\begin{align*}
&n(R_{12}+R_{34}-\epsilon)\\
&\overset{(a)}{\leq} I(M_{12};Y_2^n|M_{21},M_{43})+I(M_{34};Y_4^n,Y_2^n|M_{12},M_{21},M_{43})\\
&\leq I(M_{12};Y_2^n|M_{21},M_{43})+I(M_{34};Y_2^n|M_{21},M_{12},M_{43})+H(Y_4^n|M_{21},M_{12},M_{43},Y_2^n)\\
&\overset{(b)}{=}I(M_{12};Y_2^n|M_{21},M_{43})+I(M_{34};Y_2^n|M_{21},M_{12},M_{43})\\
&+\sum_{i=1}^n[H(S^{N-n_{34}}X_{3,i}|M_{21},M_{12},M_{43},Y_4^{i-1},X_4^i,Y_2^n,X_2^n,X_1^i)]\\
&\leq \sum_{i-1}^n [H(Y_{2,i}|Y_2^{i-1},M_{21},X_2^i)-H(Y_{2,i}|Y_2^{i-1},M_{12},M_{21},M_{43})+H(Y_{2,i}|Y_2^{i-1},M_{12},M_{21},M_{43})\\
&+H(S^{N-n_{34}}X_{3,i}|M_{21},M_{12},M_{43},Y_4^{i-1},X_4^i,S^{N-n_{12}}X_1^n+S^{N-n_{22}}X_2^n+S^{N-n_{32}}X_3^n,X_2^n,X_1^i)]\\
&\leq \sum_{i=1}^n [H(S^{N-n_{12}}X_{1,i}+S^{N-n_{32}}X_{3,i})+H(S^{N-n_{34}}X_{3,i}|S^{N-n_{32}}X_{3,i})]\\
&\leq n(\max (n_{12},n_{32})+[n_{34}-n_{32}]^+)\\
&\overset{(c)}{=} n(\max (p,q)+[p-q]^+).
\end{align*}

We introduce the genie $Y_2^n$ in the second mutual information term in (a), i.e. we provide asymmetric side information to only one receiver. In (b), we add $X_1^i$ in the entropy term because of the iterated argument that, given $M_{12},X_2^n,X_4^i$, we can construct $X_1^i$. For (c), we assumed a symmetric channel. 



\begin{remark} 
Note that we do {\bf not} need partial adaptation in this bound, and so these conclusions actually hold for {\bf full adaptation.} This implies that for the symmetric channel, full adaptation is useless when two-way interference is strong ($1\leq \alpha\leq 2, \alpha=q/p$) and weak in some interval ($2/3\leq \alpha\leq 1, \alpha=q/p$) where this outer bound may be achieved. Interestingly, when $2/3\leq \alpha\leq 2$, the ``V'' curve is also the capacity for the linear deterministic symmetric interference channel with  feedback \cite{Changho2010}.
If we add another asymmetric genie $Y_4^n$ in the first term in (a), then we obtain the second sum-rate bound (IC$\rightarrow$ c). 
\end{remark}
 
It may further be shown that for symmetric channels,  adaptation is also useless when two-way interference is very strong ($\alpha>2, \alpha=q/p$). To show this, we re-derive the single-rate bounds this time {\it not} assuming partial adaptation (allowing for full adaptation), and using symmetry in the last step:
\begin{align*}
 n(R_{12}-\epsilon)
& \leq I(M_{12};Y_2^n,Y_3^n|M_{21},M_{34})\\
& \leq H(Y_2^n,Y_3^n|M_{21},M_{34})\\
& = \sum_{i=1}^n [H(Y_{2,i},Y_{3,i}|Y_2^{i-1},Y_3^{i-1},M_{21},M_{34},X_2^i,X_3^i)]\\
& = \sum_{i=1}^n [H(S^{N-n_{12}}X_{1,i},S^{N-n_{43}}X_{4,i}|Y_2^{i-1},Y_3^{i-1},M_{21},M_{34},X_2^i,X_3^i)]\\
& \leq \sum_{i=1}^n [H(S^{N-n_{12}}X_{1,i},S^{N-n_{43}}X_{4,i})]\\
& = n\max (n_{12},n_{43})\\
&=np
\end{align*}
Under very strong interference constraints,  this is also known to be achievable.  Thus, we have obtained the capacity  for the  symmetric linear deterministic two-way IC when $\alpha\geq 2/3$, where we see that full adaptation is useless.  We will comment more on this in Remark \ref{curves}, and in Fig. \ref{fig:curves}.

{We now continue with the sum-rate outer bound (IC$\rightarrow$ d), which uses a similar genie to that in  Costa and El Gamal's \cite{elgamal_det_IC} capacity result for a class of deterministic ICs, i.e. gives to one receiver the interference created at the {\it other} receiver by the desired message. The same type of genie (though this time noisy) is used in the new outer bound for the Gaussian one-way interference channel by Etkin, Tse and Wang \cite{etkin_tse_wang}. The main difference is that we also provide the transmitters in the $\rightarrow$ direction the messages in the $\leftarrow$ direction, or $M_{21}$ and $M_{43}$, in order to be able to create $X_2$ and $X_4$ and remove these from the entropy terms, obtaining only entropies of combinations of the variables in the $\rightarrow$ direction (of $X_1, X_3$) for the sum-rate bound for the $\rightarrow$ direction.}

 \begin{align*}
&n(R_{12}+R_{34}-\epsilon)\leq I(M_{12};Y_2^n,S^{N-n_{14}}X_1^n,M_{21},M_{43})+I(M_{34};Y_4^n,S^{N-n_{32}}X_3^n,M_{21},M_{43})\\
& \overset{(d)}{=} H(Y_2^n|S^{N-n_{14}}X_1^n,M_{43},M_{21})+H(S^{N-n_{14}}X_1^n|M_{43},M_{21})-H(Y_2^n,S^{N-n_{14}}X_1^n|M_{12},M_{21},M_{43})\\
& +H(Y_4^n|S^{N-n_{32}}X_3^n,M_{43},M_{21})+H(S^{N-n_{32}}X_3^n|M_{43},M_{21})-H(Y_4^n,S^{N-n_{32}}X_3^n|M_{34},M_{21},M_{43})\\
& \overset{(e)}{=} H(Y_2^n|S^{N-n_{14}}X_1^n,M_{43},M_{21})+H(Y_4^n|S^{N-n_{32}}X_3^n,M_{43},M_{21}) \\
& + \sum_{i=1}^n H(S^{N-n_{14}}X_{1,i}|S^{N-n_{14}}X_{1}^{i-1},M_{43},M_{21}, M_{34})-H(Y_{2,i}|Y_2^{i-1}, M_{12},M_{21},M_{43}, X_2^i, X_1^i)\\
& + \sum_{i=1}^n H(S^{N-n_{32}}X_{3,i}|S^{N-n_{32}}X_{3}^{i-1}, M_{43},M_{21}, M_{12})-H(Y_{4,i}|Y_4^{i-1}, M_{34},M_{21},M_{43}, X_4^i, X_3^i)\\
& \overset{(f)}{=} H(Y_2^n|S^{N-n_{14}}X_1^n,M_{43},M_{21})+H(Y_4^n|S^{N-n_{32}}X_3^n,M_{43},M_{21}) \\
& + \sum_{i=1}^n H(S^{N-n_{14}}X_{1,i}|S^{N-n_{14}}X_{1}^{i-1},M_{43},M_{21}, M_{34}, X_3^i, X_4^i)-H(S^{N-n_{32}}X_{3,i}|S^{N-n_{32}}X_3^{i-1}, M_{12},M_{21},M_{43},  X_2^i, X_1^i)\\
& + \sum_{i=1}^n H(S^{N-n_{32}}X_{3,i}|S^{N-n_{32}}X_{3}^{i-1}, M_{43},M_{21}, M_{12}, X_1^i, X_2^i)-H(S^{N-n_{14}}X_{1,i}|S^{N-n_{14}}X_1^{i-1},  M_{34},M_{21},M_{43}, X_4^i, X_3^i)\\
& = H(Y_2^n|S^{N-n_{14}}X_1^n,M_{43},M_{21})+H(Y_4^n|S^{N-n_{32}}X_3^n,M_{43},M_{21})\\
& \leq \sum_{i=1}^n [H(S^{N-n_{12}}X_{1,i}+S^{N-n_{32}}X_{3,i}|S^{N-n_{14}}X_{1,i})+H(S^{N-n_{14}}X_{1,i}+S^{N-n_{34}}X_{3,i}|S^{N-n_{32}}X_{3,i})]\\
& \leq n(\max ([n_{12}-n_{14}]^+,n_{32})+\max ([n_{34}-n_{32}]^+,n_{14})),
\end{align*}
where (d) follows from the independence of the messages, (e)  by the chain rule of entropy and the fact that we can create $X_1^i$ given $M_{12}$, and we can create $X_2^i$ given $M_{21}, Y_2^{i-1}$ (similarly for $X_3^i, X_4^i$). We have also added the independent messages $M_{34}$ and $M_{12}$ to the 3rd and 5th term conditioning. For (f), we have expanded the $Y_{2,i}$ and $Y_{4,i}$ in the entropy terms of the 4th and 6th terms and removed the contributions from the conditioning. 
In the 3rd term, we can create $X_3^i$ from $M_{34}$ and $X_4^i$ from $S^{N-n_{14}}X_1^{i-1}$, $X_3^i$ and $M_{43}$ (similarly for the 5th term creating $X_1^i$ and $X_2^i$). 

The sum-rate bound in the opposite direction (which we must consider given the fact that under partial adaptation, not everything is symmetric):
\begin{align*}
&n(R_{21}+R_{43}-\epsilon) \leq I(M_{21};Y_1^n,S^{N-n_{23}}X_2^n,M_{12},M_{34})+I(M_{43};Y_3^n,S^{N-n_{41}}X_4^n,M_{12},M_{34})\\
& \overset{(f)}{=}H(Y_1^n|S^{N-n_{23}}X_2^n,M_{12},M_{34})+H(S^{N-n_{23}}X_2^n|M_{12},M_{34})-H(Y_1^n,S^{N-n_{23}}X_2^n|M_{12},M_{34},M_{21})\\
& +H(Y_3^n|S^{N-n_{41}}X_4^n,M_{12},M_{34})+H(S^{N-n_{41}}X_4^n|M_{12},M_{34})-H(Y_3^n,S^{N-n_{41}}X_4^n|M_{12},M_{34},M_{43})\\
& \overset{(g)}{=}H(Y_1^n|S^{N-n_{23}}X_2^n,M_{12},M_{34})+H(S^{N-n_{23}}X_2^n|M_{12},M_{34},M_{43})-H(S^{N-n_{41}}X_4^n|M_{12},M_{34},M_{21})\\
& +H(Y_3^n|S^{N-n_{41}}X_4^n,M_{12},M_{34})+H(S^{N-n_{41}}X_4^n|M_{12},M_{34},M_{21})-H(S^{N-n_{23}}X_2^n|M_{12},M_{34},M_{43})\\
& \leq \sum_{i=1}^n [H(S^{N-n_{21}}X_{2,i}+S^{N-n_{41}}X_{4,i}|S^{N-n_{23}}X_{2,i})+H(S^{N-n_{43}}X_{4,i}+S^{N-n_{23}}X_{2,i}|S^{N-n_{41}}X_{4,i})]\\
& \leq n(\max ([n_{21}-n_{23}]^+,n_{41})+\max ([n_{43}-n_{41}]^+,n_{23})),
\end{align*}

where (f) follows from the independence of the messages. Equation (g) follows since $X_1$ and $X_3$ are functions only of $M_{12}$ and $M_{34}$ and from Lemma \ref{lemma:partial}.

\begin{remark} We needed partial adaptation (Lemma \ref{lemma:partial}) in the proof of the previous two bounds (IC$\rightarrow$ d) and (IC$\leftarrow$ d). 
 In the above, nodes 1 and 3 were restricted. 
 By symmetry, we may obtain the same result if nodes 2 and 4 were restricted.
 
 Finally, 
\end{remark}
\begin{align*}
& n(2R_{12}+R_{34}-\epsilon)\\
& \leq I(M_{12};Y_2^n|M_{21},M_{43})+I(M_{12};Y_2^n,Y_4^n|M_{21},M_{43},M_{34})+I(M_{34};Y_4^n,S^{N-n_{32}}X_3^n|M_{21},M_{43})\\
&=H(Y_2^n|M_{21},M_{43})-H(Y_2^n|M_{21},M_{43},M_{12})+H(Y_4^n|M_{21},M_{43},M_{34})\\
& +H(Y_2^n|M_{21},M_{43},M_{34},Y_4^n)+H(Y_4^n,S^{N-n_{32}}X_3^n|M_{21},M_{43})-H(Y_4^n,S^{N-n_{32}}X_3^n|M_{34},M_{21},M_{43})\\
& \overset{(h)}{=}H(Y_2^n|M_{21},M_{43})-H(Y_2^n|M_{21},M_{43},M_{12})+H(S^{N-n_{32}}X_3^n|M_{43},M_{21},M_{12})\\
& +H(Y_4^n|S^{N-n_{32}}X_3^n,M_{43},M_{21})+H(Y_4^n|M_{21},M_{43},M_{34})\\
& -H(Y_4^n,S^{N-n_{32}}X_3^n|M_{34},M_{21},M_{43})+H(Y_2^n|M_{21},M_{43},M_{34},Y_4^n)\\
& \overset{(i)}{\leq} \sum_{i=1}^n [H(S^{N-n_{12}}X_{1,i}+S^{N-n_{32}}X_{3,i})+H(S^{N-n_{14}}X_{1,i}+S^{N-n{34}}X_{3,i}|S^{N-n_{32}}X_{3,i})+H(S^{N-n_{12}}X_{1,i}|S^{N-n_{14}}X_{1,i})]\\
& =n(\max (n_{12},n_{32})+\max ([n_{34}-n_{32}]^+,n_{14})+[n_{12}-n_{14}]^+),
\end{align*}
where  (h) follows from the definition of partial adaptation and Lemma \ref{lemma:partial} (skipping a transition to multi-letter for brevity), and (i) by canceling the 2nd and 3rd terms, as well as the 5th and the 6th terms.
We may similarly prove the other bounds of this form (IC$\rightarrow$ f), (IC$\leftarrow$ e) and (IC$\leftarrow$ f).\end{proof}

We again see that, under partial adaptation constraints, adaptation is useless and we obtain the capacity region of two one-way ICs. Essentially, {\it partial} adaptation prevented messages being relayed by other messages (which was also impossible in the MAC/BC and Z channels). For example,  under full adaptation, message $M_{12}$ may be relayed from Tx1 to Rx 2 through nodes 3 and 4. This path is ``blocked'' by the partial adaptation assumption,  as node 3 could not adapt to carry $M_{12}$. However, it should be pointed out that this is not necessary in general: full adaptation in the two-way modulo 2 adder IC is useless as we showed in the previous subsection, but the path is not blocked.

\subsection{Symmetric rate comparison with other interference channel models} 
For symmetric deterministic linear ICs, we may compare the symmetric sum-capacity $C_{sym}$ of various one-way and two-way models. 
Recalling $\alpha : = q/p$, we plot $C_{sym}$ as a function of $\alpha$ for the IC \cite{etkin_tse_wang}, the IC with noiseless output feedback \cite{Changho2010}, the IC with rate-limited feedback \cite{Vahid-IC-rate-FB} (for a fixed value of $\beta = 0.125$ in the notation of \cite{Vahid-IC-rate-FB}), and the two-way IC with full adaptation considered here (for $\alpha\geq 2/3$ only). Several observations may be made: the two-way IC with {\it partial} adaptation behaves like  two one-way interference channels operating in parallel over the forward and backwards link, as seen by the coinciding lines for the one-way and two-way IC with partial adaptation. This tells us that allowing partial adaptation is useless -- i.e. may as well not adapt. Interestingly, the same holds true even for full adaptation for $\alpha>2/3$. This was also concluded for the linear deterministic one-way interference channel with interfering feedback links in \cite{sahai2009channel}; what is interesting is that we can just as well squeeze in extra information messages in the feedback link (in the two-way interference channel model) rather than use the backwards links for feedback.  The symmetric sum-capacity for the {\it fully} adaptive two-way IC remains open for $\alpha<2/3$; it is solved for partial adaptation.

Recently, the work in \cite{Suh:ISIT2012, Suh2012} has considered a one-way interference channel with interfering feedback links (again forming an interference channel), a generalization of some of the deterministic interference channels with feedback considered in \cite{sahai2009channel}, where the feedback link spends fraction $\lambda$ of its time sending feedback, and uses the remaining $(1-\lambda)$ for other things (such as for example sending independent backwards messages, though adaptation as in \eqref{eq:adapt} is not considered). This is quite different from our model which integrates sending feedback and messages over all links, allows for adaptation, and does not force this separation.  While the symmetric sum-capacity for this two-message interference channel with interfering feedback links is obtained in \cite{Suh:ISIT2012, Suh2012} in our notation for $\alpha\geq 2/3$, it is a function of this parameter $\lambda$ and is thus not plotted here. 
\begin{figure*}
\centerline{\includegraphics[width=12cm]{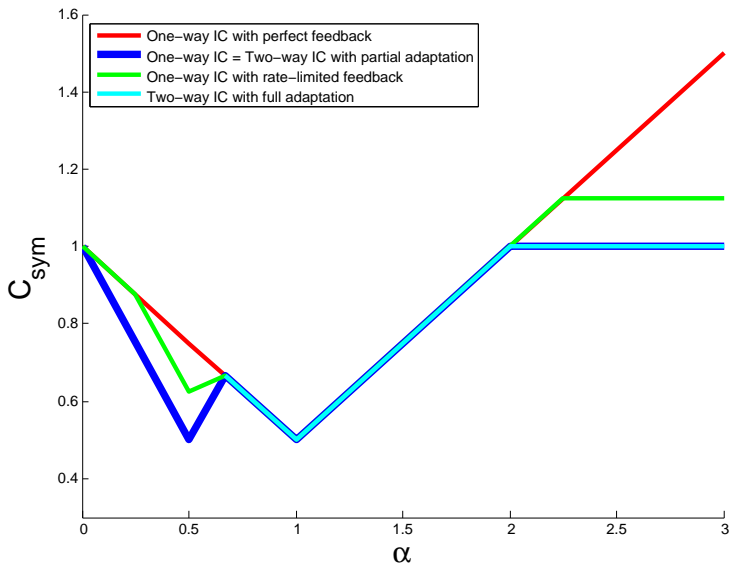}}
\vspace{-3.9cm}
\caption{$C_{sym}$ for various linear deterministic ICs as a function of $\alpha:= \frac{q}{p}$; $q$ interfering link strength, $p$ direct link strength.}
\label{fig:curves}
\end{figure*}
\label{curves}

\section{Gaussian Two-way Multiple-Access Broadcast Channel}
\label{MACBCG}

All previous channel models considered were deterministic. We now ask whether we may obtain insight into whether adaptation is useless / useful in certain noisy channels. We do so by considering the Gaussian two-way MAC/BC in this section, and the Gaussian two-way IC in the next.

We demonstrate that adaptation in the real Gaussian two-way MAC/BC can only improve the sum-capacity up to $1/2$ bit per direction. We show this by comparing non-adaptive inner bounds for this channel to outer bounds to the two-way Gaussian MAC/BC. Our outer bound for the $\rightarrow$ MAC direction is derived directly; the outer bound for the $\leftarrow$ BC direction 
follows by enhancing the BC channel by giving Tx 2 perfect output feedback and rendering the channel degraded, at which point the converse of \cite[Thm.2]{elgamal_deg_BC} follows.  

\subsection{Channel model}
At each channel use, the Gaussian two-way MAC/BC is described by the input/output relationships
\begin{align*}
&Y_1=X_2+Z_1\\
&Y_2=X_1+X_3+Z_2\\
&Y_3=X_2+Z_3,
\end{align*}
subject to power constraints $E[|X_j|^2]\leq P_j, j\in\{1,2,3\}$, and independent, identically distributed complex Gaussian noise $Z_j\sim \mathcal{CN}(0, N_j)$ at all nodes  $j\in (1,2,3)$, which may be done without loss of generality (notice both the arbitrary power and noise constraints). WLOG assume that $N_3\geq N_1$. 
 Note that we have removed the ``self-interference'' terms such as $X_1$ in the expression of $Y_1$ (for example) in contrast to the deterministic models considered. This is for ease of exposition, to make the parallels with the MAC and BC channels more direct. Note that in a Gaussian model, these ``self-interference'' terms can always be subtracted at a given node in any case. To contrast, in the Gaussian interference channel next, we will NOT eliminate the self-interference terms from the channel model, to demonstrate how they may be handled directly. 
Recall that the inputs of the Gaussian two-way MAC/BC are fully adaptive, i.e. 
\begin{align}
& X_{1,i}=f_1(M_{12},Y_1^{i-1}), \;\; X_{2,i}=f_2(M_{21}, M_{23}, Y_2^{i-1}), \;\;  X_{3,i}=f_3(M_{32},Y_3^{i-1}).
\end{align}

\subsection{The limited utility of adaptation in the Gaussian two-way MAC/BC}
We now have the following theorem. 

\medskip
\begin{theorem}
Adaptation in the Gaussian MAC/BC channel may only improve the sum-rate in the $\rightarrow$ and $\leftarrow$ directions by up to $1/2$ bit per direction. 
\label{thm:MACBCG}
\end{theorem}
\begin{proof}
First let us consider the $\rightarrow$ direction. 
For achievability, let the $\rightarrow$ direction use the capacity achieving scheme for the non-adaptive Gaussian multiple access channel, whose sum-rate is  dominated by 
\begin{equation}
R_{12}+R_{32} \leq \frac{1}{2}\log\left(1+\frac{P_1+P_3}{N_2}\right)
\label{eq:MAC}
\end{equation}

For the converse, first consider the MAC direction, and follow steps along the lines of a MAC with feedback as in \cite{Ozarow:MAC-FB, willems_FB}, which are however not immediately applicable:
\begin{align}
n(R_{12}+R_{32}) & = H(M_{12}, M_{32}) = H(M_{12}, M_{32}|M_{21}, M_{23}) \\
& = H(M_{12}, M_{32}|M_{21}, M_{23}, Y_2^n) + I(M_{12}, M_{32};Y_2^n|M_{21}, M_{23}) \\
&\overset{(a)}{\leq} n\epsilon_n + \sum_{i=1}^n H(Y_{2,i}|Y_2^{i-1}, M_{21}, M_{23}) - H(Y_{2,i}|Y_2^{i-1}, M_{12}, M_{32}, M_{21}, M_{23}, Y_1^n, Y_3^n) \\
& \overset{(b)}{\leq}  n\epsilon_n + \sum_{i=1}^n H(Y_{2,i}|X_{2,i}) - H(Y_{2,i}|X_{1,i}, X_{2,i}, X_{3,i}) \\
& \overset{(c)}{\leq} n\epsilon_n + \sum_{i=1}^n H(X_{1,i}+X_{3,i}+Z_{2,i}) - H(Z_{2,i}) \\
& \overset{(d)}{\leq} n\epsilon_n + \frac{n}{2}\log\left(1+\frac{P_1+P_3+2\sqrt{P_1P_3}}{N_2}\right)  \label{eq:MACG2}
\end{align}
where (a) follows by Fano's inequality for the first term, by the chain rule of entropy for the 2nd and 3rd terms, and by conditioning reduces entropy in adding $Y_1^n$ and $Y_3^n$ to the 3rd term, (b) since for the 2nd term, given $M_{21}, M_{23}$ and $Y_2^{i-1}$ one can construct $X_{2,i}$ and then conditioning reduces entropy, and for the 3rd term since given all the terms in the conditioning we may create $X_{1,i}, X_{2,i}$, and $X_{3,i}$ and then use the memoryless property of the channel model, (c) follows by conditioning reduces entropy and by the memoryless channel, (d) since it suffices to consider $X_1,X_3$ to be jointly Gaussian and is outer bounded when they are maximally correlated, as adaptation may permit joint $p(x_1,x_3)$. 



Now, taking the difference between the outer bound to the adaptive two-way MAC/BC in the MAC direction in \eqref{eq:MACG2} and the non-adaptive inner bound of \eqref{eq:MAC} yields
\begin{align*}
\eqref{eq:MACG2} - \eqref{eq:MAC} & = \frac{1}{2}\log\left(\frac{N_2+P_1+P_3+2\sqrt{P_1 P_3}}{N_2+P_1+P_3} \right)\\
& = \frac{1}{2}\log\left(1+ \frac{2\sqrt{P_1 P_3}}{N_2+P_1+P_3} \right)\\
& \leq \frac{1}{2}\log\left(1+ \frac{P_1+P_3}{P_1+P_3} \right) = \frac{1}{2},
\end{align*}
where the inequality follows as $(\sqrt{P_1}-\sqrt{P_3})^2 \geq 0$ implies ${2}\sqrt{P_1P_3}\leq P_1 + P_3$, and we have decreased the denominator.

For the $\leftarrow$ direction use the capacity achieving scheme for the non-adaptive single-antenna Gaussian broadcast channel, which yields the rates, for $0\leq \alpha\leq 1$
\begin{equation}
R_{21} \leq \frac{1}{2}\log\left(1+\frac{\alpha P_2}{N_1}\right), \;\; R_{23} \leq \frac{1}{2}\log\left(1+\frac{(1-\alpha)P_2}{N_3+\alpha P_2}\right). \label{eq:BCin}
\end{equation}

For the converse, for the BC $\leftarrow$ direction we enhance the channel as follows:
\begin{itemize}
\item Give Tx 2 perfect output feedback, i.e. access to $Y_{1,i-1}$, $Y_{3,i-1}$ at time $i$ as well as access to $M_{12}, M_{32}$.  Together with feedback, this allows it to create $X_{1}^i, X_{3}^i$. 
\item Render the channel physically degraded by providing Rx 1 with $Y_3^n$. Then Rx 3's output is trivially a physically degraded version of Rx 1's output. This is where we use the fact that, WLOG $N_3\geq N_1$ (if the reverse had been true we would have given $Y_1^n$ to Rx 3 instead). This is crucial in ensuring a constant gap to a non-adaptive scheme. 
\end{itemize}
The converse of \cite[Thm. 2]{elgamal_deg_BC}, which shows that feedback does not change the capacity region of the physically degraded BC,  then follows along all the same steps with the notation correspondences \break (\cite{elgamal_deg_BC} $\leftrightarrow$ this paper) as follows: 
\[ W_1 \leftrightarrow M_{21}, \; W_2 \leftrightarrow M_{23}, \; Y^n \leftrightarrow (Y_1^n, Y_3^n), \; Z^n \leftrightarrow Y_3^n, \; X^n \leftrightarrow X_2^n \]
The key point in proving the converse is \cite[Lemma 3]{elgamal_deg_BC}, which follows in a straightforward manner even given the added adaptation constraint (i.e. $X_{2,i}$ is also a function of $Y_2^{i-1}$ which is not present in the original \cite[Thm. 2]{elgamal_deg_BC}), but we re-state and prove it here in our notation for clarity and completeness.
\begin{lemma}{\it Analogous to Lemma 3 of \cite{elgamal_deg_BC}.} For all $\lambda\geq 0$,
\begin{align*}
n(R_{21}+R_{23}) &\leq I(M_{23};Y_3^n)+\lambda I(M_{21};Y_1^n,Y_3^n|M_{23}) \\
& \leq \sum_{i=1}^n I(U_i;Y_{3,i})+\lambda I(X_{2,i};Y_{1,i},Y_{3,i}|U_i) 
\end{align*}
where $U_i := (M_{23}, Y_1^{i-1}, Y_3^{i-1})$.
\end{lemma}
\begin{proof}
For the first term,
\begin{align}
I(M_{23};Y_3^n) &= \sum_{i=1}^n I(M_{23};Y_{3,i}|Y_3^{i-1}) \\
& = \sum_{i=1}^n H(Y_{3,i}|Y_3^{i-1}) - H(Y_{3,i}|M_{23}, Y_3^{i-1}) \\
&\leq \sum_{i=1}^n H(Y_{3,i}) - H(Y_{3,i}|M_{23}, Y_3^{i-1}, Y_1^{i-1}) \\
& = \sum_{i=1}^n I(Y_{3,i};U_i)
\end{align}
by definition of $U_i : = (M_{23}, Y_3^{i-1}, Y_1^{i-1})$. For the second term,
\begin{align}
I(M_{21};Y_1^n, Y_3^n|M_{23}) & = \sum_{i=1}^n I(M_{21};Y_{1,i}, Y_{3,i}|M_{23}, Y_1^{i-1}, Y_3^{i-1}) \\
& = \sum_{i=1}^n I(M_{21};Y_{1,i}, Y_{3,i}|U_i) \\
&\leq \sum_{i=1}^n I(M_{21}, X_{2,i};Y_{1,i}, Y_{3,i}|U_i) \\
& = \sum_{i=1}^n H(Y_{1,i}, Y_{3,i}|U_i) - H(Y_{1,i}, Y_{3,i}|U_i, M_{21}, X_{2,i})\\
&=\sum_{i=1}^n I(Y_{1,i}, Y_{3,i};X_{2,i}|U_i)
\end{align}
where several steps in the proof of \cite[Lemma 3]{elgamal_deg_BC} are not needed as our channel is trivially degraded. 
\end{proof}
Following the same arguments as in \cite[Thm. 2]{elgamal_deg_BC}, the above Lemma yields an outer bound equivalent to the region in \eqref{eq:BCout}, where we note that in addition to $U_i = (M_{23}, Y_1^{i-1}, Y_3^{i-1})$ to construct $X_{2,i} = f(M_{21}, M_{23}, Y_2^{i-1}, M_{12}, M_{32}, Y_1^{i-1}, Y_3^{i-1}) \equiv f(U_i, M_{21}, Z_2^{i-1}, M_{12}, M_{32})$ we also need $M_{12}, M_{32}, M_{21}, Z_2^n$, but that, given the above definition of the random variable $U_i$, the factorization of the inputs as $p(u)p(x_2|u)$ still holds. Note that with some abuse of notation we have left the channel distribution as $p(y_1, y_3|x_2)p(y_3|y_1,y_3)$ to emphasize that Rx 1 has access to both $Y_1^n, Y_3^n$ (we have forced the channel to be degraded) and thus that Rx 3, with access to $Y_3^n$ only is trivially a degraded version of this. The outer bound for the $\leftarrow$ BC direction is thus given by the set of all non-negative $R_{21}, R_{23}$ such that
\begin{align}
R_{21} &\leq I(X_2;Y_1, Y_3|U), \;\;\;\;\;\; R_{23} \leq I(U;Y_3) \label{eq:BCout}
\end{align}
over all distributions of the form $p(u)p(x_2|u)p(y_1,y_3|x_2)p(y_3|y_1,y_3)$.  Evaluation for the Gaussian channel, as done in \cite{Ozarow:1984}, yields an outer bound of 
\begin{align}
R_{21} \leq \frac{1}{2}\log\left(1+\frac{\alpha P_2}{N_1}\frac{N_1+N_3}{N_3}\right), \;\;\;\; R_{23} \leq \frac{1}{2}\log\left( 1+\frac{(1-\alpha)P_2}{N_3+\alpha P_2} \right) \label{BCoutG}
\end{align}
for $0\leq \alpha \leq 1$.

Taking the difference between the sum of the outer bound to the adaptive two-way MAC/BC in the BC direction in \eqref{BCoutG} and the sum of the  non-adaptive inner bounds of \eqref{eq:BCin} yields
\begin{align*}
\eqref{BCoutG} - \eqref{eq:BCin} & = \frac{1}{2}\log\left(1+\frac{\alpha P_2}{N_1}\frac{N_1+N_3}{N_3}\right) - \frac{1}{2}\log\left(1+\frac{\alpha P_2}{N_1}\right)   \\
& \overset{(a)}{\leq}\frac{1}{2}\log\left(1+\frac{2 \alpha P_2}{N_1}\right) - \frac{1}{2}\log\left(1+\frac{\alpha P_2}{N_1}\right)  \leq  \frac{1}{2}.
\end{align*}
where (a) follows as $\frac{N_1+N_3}{N_3} = 1+\frac{N_1}{N_3} \leq 2$ since $N_1\leq N_3$. 
\end{proof}

\begin{remark}
We note that this result also implies that for the one-way Gaussian MAC with FB and the one-way BC with FB, feedback and adaptation of the nodes can only increase capacity by up to 1/2 bit (sum-rate) per direction. This fact has been partially noted in \cite{Changho2010}.
\end{remark}

\section{Gaussian Two-way Interference Channel}
\label{ICG}

We now consider the Gaussian two-way interference channel, and ask when non-adaptive schemes such as the celebrated Han and Kobayashi \cite{H+K} perform as well, or nearly as well, as adaptive schemes. 

We do not construct any inner bounds which employ adaptation; our focus is on showing when non-adaptive schemes perform ``well''. Rather, we derive an outer bounds for the Gaussian two-way IC under full adaptation (all 4 nodes may adapt) and several under partial adaptation (only 2 of the 4 may adapt) constraints. We then show that non-adaptive schemes sometimes achieve the capacity, or at least to within a constant gap of either the fully or partially adaptive schemes. 
We note that while the converses and the steps are new and exploit carefully chosen genies, when we evaluate these by further outer-bounding our outer-bounds, interestingly, we  sometimes re-obtain some of the outer bounds of the interference channel \cite{etkin_tse_wang} {\it or} the interference channel with feedback \cite{Changho2010}. This  in turn is sufficient to achieve capacity to within a constant gap, which we emphasize, sometimes is limited to {\it partial} adaptation for some of the weak interference regimes but this will be explicitly mentioned when it is the case. 


\subsection{Channel model, definitions, and partial adaptation lemma}
At each channel use, the Gaussian two-way IC is described by the input/output relationships
\begin{align*}
&Y_1=g_{11}X_1+g_{21}X_2+g_{41}X_4+Z_1\\
&Y_2=g_{12}X_1+g_{22}X_2+g_{32}X_3+Z_2\\
&Y_3=g_{23}X_2+g_{33}X_3+g_{43}X_4+Z_3\\
&Y_4=g_{14}X_1+g_{34}X_3+g_{44}X_4+Z_4,
\end{align*}
where $g_{jk}$, for $j,k\in \{1,2,3,4\}$ are the complex channel gains. We assume the power constraints $E[|X_j|^2]\leq P_j=1, j\in\{1,2,3,4\}$, and independent, identically distributed complex Gaussian noise $Z_j\sim \mathcal{CN}(0, 1)$ at all nodes  $j\in (1,2,3,4)$, which may be done without loss of generality. 
Furthermore, we define ${\tt SNR}_{12}=|g_{12}|^2, {\tt SNR}_{21}=|g_{21}|^2, {\tt SNR}_{34}=|g_{34}|^2, {\tt SNR}_{43}=|g_{43}|^2$, and ${\tt INR}_{14}=|g_{14}|^2, {\tt INR}_{41}=|g_{41}|^2, {\tt INR}_{23}=|g_{23}|^2, {\tt INR}_{32}=|g_{32}|^2$. Note that we have kept the ``self-interference'' terms such as $g_{11}X_1$ in the expression of $Y_1$ (for example). In this Gaussian model, it is clear that since node $1$ knows $X_1$ we may remove this self-interference term due to the additive nature of the channel. However,  we leave it in our expressions to emphasize precisely this fact. In other channels such as the two-way binary multiplier channel, where $Y=X_1 X_2$ one cannot ``undo'' ones' own channel, which is one source of difficulty for this elusive two-way channel. In all converses, the fact that we can cancel or subtract out a node's ``self-interference'' is shown explicitly. This is one of the reasons two-way channels of this form, as seen for example in the Gaussian two-way channel as well \cite{Han:1984}, are easier to deal with, which we emphasize. 

We say that the Gaussian two-way interference channel operates under ``full adaptation'' if we allow
\begin{align}
& X_{1,i}=f_1(M_{12},Y_1^{i-1}), \;\; X_{2,i}=f_2(M_{21},Y_2^{i-1})\\
& X_{3,i}=f_3(M_{34},Y_3^{i-1}), \;\; X_{4,i}=f_4(M_{43},Y_4^{i-1}).
\end{align}
Similarly, it operates under ``partial adaptation'' if we only allow the following:  
\begin{align}
& X_{1,i}=f_1(M_{12})\label{eq:p1}, \;\; X_{2,i}=f_2(M_{21},Y_2^{i-1})\\
& X_{3,i}=f_3(M_{34}), \;\; X_{4,i}=f_4(M_{43},Y_4^{i-1}),\label{eq:p2}
\end{align}
i.e.  nodes 1 and 3 are ``restricted'' \cite{Shannon:1961}.  By symmetry, we may alternatively allow nodes $2$ and $4$ to be restricted and $1,3$ to be fully adaptive; whether allowing $1,2$ or $1,4$ to be restricted and the complement fully adaptive remains an open problem.

We are interested in the symmetric capacity (or sum-rate), when all the SNRs equal a given ${\tt SNR}$, and all the INRs equal a given ${\tt INR}$. For full adaptation, due to the symmetry,  we consider the per-user rates $R_{sym} = \frac{R_{12}+R_{34}}{2} = \frac{R_{21}+ R_{43}}{2}$.  In partial adaptation, there is only partial symmetry (nodes 1 and 3 are fixed, while 2 and 4 are not),  and hence we will consider the per user rates $R_{sym\rightarrow} = \frac{R_{12}+R_{34}}{2}$ and $R_{sym\leftarrow} = \frac{R_{21}+R_{43}}{2}$ for the forward and reverse directions respectively. We will derive outer bounds for $R_{sym}$ under full adaptation and $R_{sym\rightarrow}$, $R_{sym\leftarrow}$ under partial adaptation, and show these to be achievable to within constant gaps by non-adaptive schemes. 

We first prove a modified version of Lemma \ref{lemma:partial} relevant in partial adaptation for the Gaussian channel.
\begin{lemma}
Under partial adaptation \eqref{eq:p1} -- \eqref{eq:p2}, for some deterministic functions $f_5$ and $f_6$, 
\begin{align}
X_{2,i}& = f_5(M_{12}, M_{21}, M_{34}, Z_2^{i-1}) \perp M_{43}, \;\; \forall i \label{eq:X2}\\
X_{4,i}& = f_6(M_{43}, M_{34}, M_{12},Z_4^{i-1}) \perp M_{21}, \;\; \forall i \label{eq:X4}
\end{align}
where $\perp$ denotes independence.
\label{lemma:partialG}
\end{lemma}
\begin{proof}
Note that $X_{2,i} = f_2(M_{21}, Y_2^{i-1})$ and $Y_2^{i-1} = g_{12}X_{1}^{i-1} + g_{22}X_2^{i-1}+g_{32}X_3^{i-1}+Z_2^{i-1}$. Since $X_1^{i-1}$ and $X_3^{i-1}$ are functions only of $M_{12}$ and $M_{34}$ respectively, we may conclude that  there exists a function $f^*$ such that $X_{2,i} = f^*(M_{21}, M_{12}, M_{34}, X_2^{i-1},Z_2^{i-1})$. Iterating this argument, and noting that $X_{2,1}$ is only a function of $M_{21}$, we obtain the lemma. The result for $X_{4,i}$ follows similarly. That $X_{2,i}$ is independent of $M_{43}$ follows since $M_{43}$ is independent of all the arguments inside $f^*$. 
\end{proof}

\subsection{Outer bounds}
We now present two outer bounds for the Gaussian two-way IC under full and partial adaptation respectively. These bounds are either within a constant gap, or sufficient to show the capacity, for the different regimes. We derive general outer bounds, imposing symmetry only in the final step.

\begin{theorem} {\it Outer bound: full adaptation.}
For the Gaussian two-way symmetric IC under full adaptation, any achievable symmetric rate $R_{sym} = \frac{R_{12}+R_{34}}{2} = \frac{R_{21}+R_{43}}{2}$, achievable by each user,  satisfies, 
\begin{align}
R_{sym}
&\leq \frac{1}{2} \log \left(1+{\tt SNR}+{\tt INR}+2\sqrt{{\tt SNR}\times {\tt INR}}\right)+\frac{1}{2}\log \left(1+\frac{{\tt SNR}}{1+{\tt INR}}\right) \label{R_strong_E}
\end{align}
\label{thm:outer-full}
\end{theorem}
\begin{proof}
It is sufficient to consider $R_{12}+R_{34}$ due to symmetry. This bound is inspired by the corresponding sum-rate bound in the linear deterministic model, i.e., we add asymmetric genie $Y_2^n$ at node 4 and this resembles the bounding technique used by Suh and Tse for the interference channel with feedback \cite{Changho2010}. We also note that we could have equivalently provided node 4 with the genie $g_{32}X_3^n+Z_2^n$ instead of $Y_2^n$ (in addition to $M_{12},  M_{21}, M_{43}$ and $Z_1^n$) which more resembles the type of genie seen in ICs and ICs with feedback. We have given $Y_2^n$ as it is then easier to see how node 4 may create $X_2^n$ based on $M_{21}$ and $Y_2^n$, and the bounds work out to the same. 
Notice the genie $Z_1^n$ in the conditioning of both terms as well which is not seen in the feedback bounds \cite{Changho2010}; this is needed in order to, together with the genie $M_{12}, M_{21}, M_{43}, Y_2^n$, be able to create $X_1^n$ at node 4 (essentially, to create $Y_1^n$ to create $X_1^n$). 
\begin{align}
&n(R_{12}+R_{34}-\epsilon) \notag\\
&\leq I(M_{12};Y_2^n|M_{21},M_{43},Z_1^n)+I(M_{34};Y_4^n,Y_2^n|M_{12},M_{21},M_{43},Z_1^n)\notag\\
&\leq I(M_{12};Y_2^n|M_{21},M_{43},Z_1^n)+I(M_{34};Y_2^n|M_{21},M_{12},M_{43},Z_1^n)+H(Y_4^n|M_{21},M_{12},M_{43},Y_2^n,Z_1^n)-H(Z_4^n)\notag\\
&\overset{(a)}{=}I(M_{12};Y_2^n|M_{21},M_{43},Z_1^n)+I(M_{34};Y_2^n|M_{21},M_{12},M_{43},Z_1^n)\notag\\
&+\sum_{i=1}^n[H(g_{34}X_{3,i}+Z_{4,i}|M_{21},M_{12},M_{43},Y_4^{i-1},X_4^i,Y_2^n,X_2^n,Z_1^n,X_1^i)]-H(Z_4^n)\notag\\
&\overset{(b)}{\leq} \sum_{i=1}^n [H(Y_{2,i}|Y_2^{i-1},M_{21},X_{2,i})-H(Y_{2,i}|Y_2^{i-1},M_{12},M_{21},M_{43},Z_1^n)\notag\\
& \ \ +H(Y_{2,i}|Y_2^{i-1},M_{12},M_{21},M_{43},Z_1^n)-H(Z_{2,i})+H(g_{34}X_{3,i}+Z_{4,i}|X_{4,i}, g_{32}X_{3,i}+Z_{2,i},X_1^i,X_2^n)-H(Z_{4,i})]\notag\\
&\overset{(c)}{\leq} \sum_{i=1}^nH(g_{12}X_{1,i}+g_{32}X_{3,i}+Z_{2,i}|X_{2,i})-H(Z_{2,i})+H(g_{34}X_{3,i}+Z_{4,i}|X_{4,i}, g_{32}X_{3,i}+Z_{2,i})-H(Z_{4,i}) \label{(c)}
\end{align}
In step (a), $X_1^i$ in the conditioning of the third term is constructed from ($M_{12},X_2^n,X_4^i,Z_1^n$). In step (b), we used conditioning reduces entropy, the second and the third term cancelled each other and $g_{32}X_{3,i}+Z_{2,i}$ in the conditioning of the fifth term is decoded from $Y_2^n$. In step (c), we only keep the self-interference $X_{4,i}$ and drop the terms  $X_1^i, X_2^n$ in the conditioning of the third term. We could leave these and express the outer bound in terms of correlation coefficients between the inputs (which in general may be correlated due to full adaptation). However, in subsequent steps we will seek to maximize, or outer bound this outer bound to obtain a simple analytical expression, which amounts to setting certain correlation coefficients to $0$, or equivalently, dropping the   terms  $X_1^i, X_2^n$ in the conditioning. 
Further evaluation yields \eqref{R_strong_E}, for details please refer to Appendix \ref{Gsum}.

\end{proof}

\begin{remark} 
{\it Sum-rate bound:} Note that the final, evaluated symmetric, normalized sum-rate bound in \eqref{R_strong_E} has the same form as the IC with perfect output feedback outer bound \cite[upper bound on (7)]{Changho2010}, though they are arrived at using different genies (though similar in many senses as mentioned above). In both channel models, inputs may be arbitrarily correlated as no additional arguments for restricting the input distributions have been made, leading to similar bounding techniques using correlation coefficients. 

\end{remark}

\bigskip

\begin{theorem} {\it Outer bound: partial adaptation.}
For the Gaussian two-way IC under partial adaptation \eqref{eq:p1} -- \eqref{eq:p2}, in addition to the bounds in Theorem \ref{thm:outer-full}, we may also conclude that any achievable rates ($R_{12},R_{21},R_{34},R_{43}$), and $R_{sym\rightarrow} = \frac{R_{12}+R_{34}}{2}$ and $R_{sym\leftarrow} = \frac{R_{21}+R_{43}}{2}$ must satisfy, 
\begin{align}
&R_{12}\leq \log (1+{\tt SNR}_{12})\label{partial_single_rate1}\\
&R_{21}\leq \log (1+{\tt SNR}_{21})\\
&R_{34}\leq \log (1+{\tt SNR}_{34})\\
&R_{43}\leq \log (1+{\tt SNR}_{43})\label{partial_singe-rate4}
\end{align}
\begin{align}
R_{sym\rightarrow} 
 & \leq \log \left( 1+{\tt INR}+{\tt SNR}-\frac{{\tt INR}\times {\tt SNR}}{1+{\tt INR}}\right) \label{R_weak_E} 
 \end{align}
 \begin{equation}   R_{sym\leftarrow} \leq \left\{ \begin{array}{l}
\log \left(1+{\tt INR}+\frac{{\tt SNR}}{{\tt INR}}\right), \ \ \mbox{if} \ {\tt SNR}\leq {\tt INR}^3\\
\log \left(1+\frac{(\sqrt{{\tt SNR}}+\sqrt{{\tt INR}})^2}{1+{\tt INR}}\right), \ \ \mbox{if} \ {\tt SNR}>{\tt INR}^3
 \end{array}\right. \label{R_weak_BE}
\end{equation}
\label{thm:outer-partial}
\end{theorem}
\begin{proof}
For the single-rate bounds, it is sufficient to show the first two due to symmetry. Notice that we must treat the $\rightarrow$ and $\leftarrow$ directions separately however due to the asymmetry of the partial adaptation problem definition.
\begin{align*}
n(R_{12}-\epsilon)&\leq I(M_{12};Y_2^n|M_{21},M_{34})\\
&\leq H(Y_2^n|M_{21},M_{34})-H(Y_2^n|M_{21},M_{34},M_{12},X_1^n,X_2^n,X_3^n)\\
&\overset{(a)}{\leq} \sum_{i=1}^n [H(Y_{2,i}|Y_2^{i-1},M_{21},X_{2,i},M_{34},X_{3,i})-H(Z_{2,i})]\\
&\leq \sum_{i=1}^n [H(g_{12}X_{1,i}+Z_{2,i})-H(Z_{2,i})]\\
&\leq \sum_{i=1}^n [\log (1+{\tt SNR}_{12})] \\
 n(R_{21}-\epsilon)&\leq I(M_{21};Y_1^n|M_{12},M_{43},M_{34},Z_4^{n-1})\\
& \leq H(Y_1^n|M_{12},M_{34},M_{43},Z_4^{n-1})- H(Y_1^n|M_{12},M_{34},M_{43},Z_4^{n-1},M_{21},X_1^n,X_2^n,X_4^n)\\
&\overset{(b)}{\leq} \sum_{i=1}^n [H(Y_{1,i}|M_{12},M_{34},M_{43},Z_4^{n-1},Y_1^{i-1},X_{1,i},X_{4,i})-H(Z_{1,i})]\\
&\leq \sum_{i=1}^n [H(g_{21}X_{2,i}+Z_{1,i})-H(Z_{1,i})]\\
&\leq \sum_{i=1}^n [\log (1+{\tt SNR}_{21})]
\end{align*}
where (a) follows from the definition of partial adaptation and (b) follows by the same reason, as well as Lemma \ref{lemma:partialG}.

%

Next, we consider the sum-rate bounds \eqref{R_weak_E} and \eqref{R_weak_BE}, which are inspired by the techniques used by Etkin, Tse and Wang for the interference channel \cite{etkin_tse_wang}.
For the $\rightarrow$ direction of the symmetric rate, 
\begin{align}
&n(R_{12}+R_{34}-\epsilon)\leq I(M_{12};Y_2^n,g_{14}X_1^n+Z_4^n,M_{21},M_{43})+I(M_{34};Y_4^n,g_{32}X_3^n+Z_2^n,M_{21},M_{43})\notag \\
& \overset{(a)}{=} H(Y_2^n|g_{14}X_1^n+Z_4^n,M_{43},M_{21})+H(g_{14}X_1^n+Z_4^n|M_{43},M_{21})-H(Y_2^n,g_{14}X_1^n+Z_4^n|M_{12},M_{21},M_{43})\notag \\
& +H(Y_4^n|g_{32}X_3^n+Z_2^n,M_{43},M_{21})+H(g_{32}X_3^n+Z_2^n|M_{43},M_{21})-H(Y_4^n,g_{32}X_3^n+Z_2^n|M_{34},M_{21},M_{43})\notag \\
& \overset{(b)}{=}H(Y_2^n|g_{14}X_1^n+Z_4^n,M_{43},M_{21})+\sum_{i=1}^n[H(g_{14}X_{1,i}+Z_{4,i}|g_{14}X_1^{i-1}+Z_4^{i-1},M_{43},M_{21},M_{34})\notag \\
&-H(Y_{2,i},g_{14}X_{1,i}+Z_{4,i}|Y_2^{i-1},g_{14}X_1^{i-1}+Z_4^{i-1},M_{12},M_{21},M_{43},X_2^i,X_1^i)]\notag \\
&+H(Y_4^n|g_{32}X_3^n+Z_2^n,M_{43},M_{21})+\sum_{i=1}^n[H(g_{32}X_{3,i}+Z_{2,i}|g_{32}X_3^{i-1}+Z_2^{i-1},M_{43},M_{21},M_{12})\notag \\
&-H(Y_{4,i},g_{32}X_{3,i}+Z_{2,i}|Y_4^{i-1},g_{32}X_3^{i-1}+Z_2^{i-1},M_{34},M_{21},M_{43},X_4^i,X_3^i)]\notag 
\end{align}
\begin{align}
& \overset{(c)}{=}H(Y_2^n|g_{14}X_1^n+Z_4^n,M_{43},M_{21})+\sum_{i=1}^n[H(g_{14}X_{1,i}+Z_{4,i}|g_{14}X_1^{i-1}+Z_4^{i-1},M_{43},M_{21},M_{34},X_3^i,Y_4^{i-1},X_4^i,\notag \\
& \ \ g_{32}X_3^{i-1}+Z_2^{i-1})-H(g_{32}X_{3,i}+Z_{2,i},Z_{4,i}|Y_2^{i-1},g_{14}X_1^{i-1}+Z_4^{i-1},M_{12},M_{21},M_{43},X_2^i,X_1^i,g_{32}X_3^{i-1}+Z_2^{i-1})]\notag \\
&+H(Y_4^n|g_{32}X_3^n+Z_2^n,M_{43},M_{21})+\sum_{i=1}^n[H(g_{32}X_{3,i}+Z_{2,i}|g_{32}X_3^{i-1}+Z_2^{i-1},M_{43},M_{21},M_{12},X_1^i,Y_2^{i-1},X_2^i,\notag \\
& \ \ g_{14}X_1^{i-1}+Z_4^{i-1})-H(g_{14}X_{1,i}+Z_{4,i},Z_{2,i}|Y_4^{i-1},g_{32}X_3^{i-1}+Z_2^{i-1},M_{34},M_{21},M_{43},X_4^i,X_3^i,g_{14}X_1^{i-1}+Z_4^{i-1})]\notag \\
& \overset{(d)}{=}\sum_{i=1}^n[H(Y_{2,i}|Y_2^{i-1},g_{14}X_1^n+Z_4^n,M_{43},M_{21})-H(Z_{2,i})+H(Y_{4,i}|Y_4^{i-1},g_{32}X_3^n+Z_2^n,M_{43},M_{21})-H(Z_{4,i})]\notag \\
& \leq \sum_{i=1}^n [H(g_{12}X_{1,i}+g_{32}X_{3,i}+Z_{2,i}|g_{14}X_{1,i}+Z_{4,i},X_{2,i})-H(Z_{2,i})\notag \\
& \ \ +H(g_{34}X_{3,i}+g_{14}X_{1,i}+Z_{4,i}|g_{32}X_{3,i}+Z_{2,i},X_{4,i})-H(Z_{4,i})] \label{last}
\end{align}
In the first step, we have given $(g_{14}X_1^n+Z_4^n)$ and $(g_{32}X_3^n+Z_2^n)$ as side information. 
Step (a) follows from the independence of the messages. In step (b), the 2nd and 5th terms follow since $g_{14}X_{1,i}$ and $g_{32}X_{3,i}$ are functions only of $M_{12}$ and $M_{34}$, and the 3rd and 6th terms follow from the definition of partial adaptation. For (c), in the conditioning of the 2nd term, we are able to add $(X_3^i,g_{32}X_3^{i-1}+Z_2^{i-1})$ due to partial adaptation constraints,  and $(Y_4^{i-1}, X_4^i)$ are constructed from $(g_{14}X_1^{i-1}+Z_4^{i-1},M_{43},X_3^i)$. The 5th term follows similarly. In step (d), $-H(Z_{2,i})$ and $-H(Z_{4,i})$ are obtained from a portion of the 6th and 3rd terms in (c) respectively using the chain rule (noises are independent from other terms), and the remainder (chain rule) of the 6th and 3rd terms are cancelled by the 2nd and 5th terms respectively.


To obtain \eqref{R_weak_E} we continue to outer bound \eqref{last} in terms of ${\tt SNR}$ and ${\tt INR}$, using the fact that Gaussians maximize entropy subject to variance constraints. Specifically, one may intuitively see that, if one defines $\lambda_{jk} = E[X_j X_k^*]$, that one may express \eqref{last} in terms of $\lambda_{12}, \lambda_{13}, \lambda_{14}, \lambda_{34}, \lambda_{23}$. One also notices from the conditional entropy expression in \eqref{last} that taking $ \lambda_{14} = \lambda_{23} = \lambda_{12}  = \lambda_{34} = 0$, and since $\lambda_{13} = 0$ (naturally, by partial adaptation) will maximize the outer bound. This may alternatively be worked out by calculating the conditional covariance matrices directly (as we will show for the next bound on $R_{\leftarrow}$). In this case then,  for each $i$, we may bound
\begin{align*}
H(g_{12}X_{1}&+g_{32}X_{3}+Z_{2}|g_{14}X_{1}+Z_{4},X_{2}) - H(Z_2)  \leq H(g_{12}X_{1}+g_{32}X_{3}+Z_{2}|g_{14}X_{1}+Z_{4}) - H(Z_2) \\
&\leq \log 2\pi e(\mbox{Var}(g_{12}X_{1}+g_{32}X_{3}+Z_{2}|g_{14}X_{1}+Z_{4}))-\log 2\pi e(\mbox{Var}(Z_2)) \\
& \leq \log\left(1+{\tt SNR}+{\tt INR} - \frac{{\tt SNR}\times {\tt INR}}{1+{\tt INR}}\right),
\end{align*}
which together with the symmetric expressions for the third and fourth terms in \eqref{last} yield \eqref{R_weak_E}.

\bigskip
For the $\leftarrow$ direction, we are similarly able to obtain:
\begin{align}
&n(R_{21}+R_{43}-\epsilon)\leq I(M_{21};Y_1^n,g_{23}X_2^n+Z_3^n,M_{12},M_{34})+I(M_{43};Y_3^n,g_{41}X_4^n+Z_1^n,M_{12},M_{34})\notag \\
& = H(Y_1^n|g_{23}X_2^n+Z_3^n,M_{34},M_{12})+H(g_{23}X_2^n+Z_3^n|M_{34},M_{12})-H(Y_1^n,g_{23}X_2^n+Z_3^n|M_{21},M_{12},M_{34})\notag \\
& +H(Y_3^n|g_{41}X_4^n+Z_1^n,M_{34},M_{12})+H(g_{41}X_4^n+Z_1^n|M_{34},M_{12})-H(Y_3^n,g_{41}X_4^n+Z_1^n|M_{43},M_{12},M_{34})\notag 
\end{align}
\begin{align}
& \overset{(a)}{=}H(Y_1^n|g_{23}X_2^n+Z_3^n,M_{34},M_{12})+\sum_{i=1}^n[H(g_{23}X_{2,i}+Z_{3,i}|g_{23}X_2^{i-1}+Z_3^{i-1},M_{34},M_{12},M_{43})\notag \\
&-H(Y_{1,i},g_{23}X_{2,i}+Z_{3,i}|Y_1^{i-1},g_{23}X_2^{i-1}+Z_3^{i-1},M_{21},M_{12},M_{34},X_1^i,Z_2^{i-1},X_2^i)]\notag \\
&+H(Y_3^n|g_{41}X_4^n+Z_1^n,M_{34},M_{12})+\sum_{i=1}^n[H(g_{41}X_{4,i}+Z_{1,i}|g_{41}X_4^{i-1}+Z_1^{i-1},M_{34},M_{12},M_{21})\notag \\
&-H(Y_{3,i},g_{41}X_{4,i}+Z_{1,i}|Y_3^{i-1},g_{41}X_4^{i-1}+Z_1^{i-1},M_{43},M_{12},M_{34},X_3^i,Z_4^{i-1},X_4^i)]\notag \\
& =H(Y_1^n|g_{23}X_2^n+Z_3^n,M_{34},M_{12})+\sum_{i=1}^n[H(g_{23}X_{2,i}+Z_{3,i}|g_{23}X_2^{i-1}+Z_3^{i-1},M_{34},M_{12},M_{43},Z_4^{i-1},g_{41}X_4^{i-1}+Z_1^{i-1},\notag \\
& \ \ X_3^i,X_4^i,Y_3^{i-1})-H(g_{41}X_{4,i}+Z_{1,i},Z_{3,i}|Y_1^{i-1},g_{23}X_2^{i-1}+Z_3^{i-1},M_{21},M_{12},M_{34},Z_2^{i-1},X_1^i,X_2^i,g_{41}X_4^{i-1}+Z_1^{i-1})]\notag \\
&+H(Y_3^n|g_{41}X_4^n+Z_1^n,M_{34},M_{12})+\sum_{i=1}^n[H(g_{41}X_{4,i}+Z_{1,i}|g_{41}X_4^{i-1}+Z_1^{i-1},M_{34},M_{12},M_{21},Z_2^{i-1},g_{23}X_2^{i-1}+Z_3^{i-1},\notag \\
& \ \ X_1^i,X_2^i,Y_1^{i-1})-H(g_{23}X_{2,i}+Z_{3,i},Z_{1,i}|Y_3^{i-1},g_{41}X_4^{i-1}+Z_1^{i-1},M_{43},M_{12},M_{34},Z_4^{i-1},X_3^i,X_4^i,g_{23}X_2^{i-1}+Z_3^{i-1})]\notag \\
& = \sum_{i=1}^n [H(Y_{1,i}|Y_1^{i-1},g_{23}X_2^n+Z_3^n,M_{34},M_{12})-H(Z_{1,i})+H(Y_{3,i}|Y_3^{i-1},g_{41}X_4^n+Z_1^n,M_{34},M_{12})-H(Z_{3,i})]\notag \\
&\leq \sum_{i=1}^n [H(g_{21}X_{2,i}+g_{41}X_{4,i}+Z_{1,i}|g_{23}X_{2,i}+Z_{3,i},X_{1,i})-H(Z_{1,i})\notag \\
& \ \ +H(g_{43}X_{4,i}+g_{23}X_{2,i}+Z_{3,i}|g_{41}X_{4,i}+Z_{1,i},X_{3,i})-H(Z_{3,i})]\label{last2}
\end{align}
The slight differences compared to the previous outer bound are: in step (a), we add $M_{43}$ in the conditioning part of the second term since it is independent of $X_{2,i}$ according to Lemma \ref{lemma:partialG}; in the conditioning part of the third term we add $Z_2^{i-1}$ (independent), which, together with $(M_{12},M_{21},M_{34})$ allow us to construct $X_2^i$. Similar arguments can be made for the fifth and sixth term. 

We again proceed to outer bound \eqref{last2} to obtain \eqref{R_weak_BE}. 
It is sufficient to evaluate the first two terms in \eqref{last2} due to symmetry. Once again, we could outer bound \eqref{last2} in terms of the conditional covariance matrices and then proceed to select values of the correlation coefficients (complex) $\lambda_{ij} : = E[X_i X_j^*]$  which maximize this outer bound. A more intuitive method is to note that  the conditional entropies in \eqref{last2} will be maximized if $\lambda_{14} = \lambda_{32}=0$, and $\lambda_{12}=\lambda_{34}=0$ (similar to \eqref{r12r34}), which may also be obtained by dropping $X_{1,i}, X_{3,i}$ in the conditioning terms. At that point, we are only left with the coefficient $\lambda_{24} = E[X_2X_4^*]$, (which in contrast to the $\rightarrow$ bound is not automatically $0$ due to the possible adaptation in the $\leftarrow$ direction. Furthermore, setting it to zero cannot be argued intuitively as we see a tradeoff.) yielding the following bound for $R_{sym \leftarrow} = \frac{R_{21}+R_{43}}{2}$ by symmetry: 
\begin{align}
R_{sym\leftarrow}&\leq H(g_{21}X_{2}+g_{41}X_{4}+Z_{1}|g_{23}X_{2}+Z_{3},X_{1})-H(Z_{1})\notag \\
&\leq H(g_{21}X_{2}+g_{41}X_{4}+Z_{1}|g_{23}X_{2}+Z_{3})-H(Z_{1})\notag \\
&\leq \log 2\pi e\left(\mbox{Var}(g_{21}X_{2}+g_{41}X_{4}+Z_{1}|g_{23}X_{2}+Z_{3})\right)-\log 2\pi e (\mbox{Var}(Z_1))\notag \\
&\leq \log \left( 1+{\tt INR}+{\tt SNR}+2|\lambda_{24}|\cos \theta\sqrt{{\tt SNR}\times {\tt INR}}-\frac{{\tt SNR}\times {\tt INR} + {\tt INR}^2|\lambda_{24}|^2 + 2\sqrt{{\tt SNR}}{\tt INR}^{3/2}|\lambda_{24}|\cos\theta}{1+{\tt INR}}\right).\label{lambda}
\end{align}
where $\theta$ is the angle of $g_{21}g_{41}^*\lambda_{24}$. To maximize \eqref{lambda}, we take the partials of the expression with respect to $|\lambda_{24}|$ and $\theta$ and set these to 0. For these to equal 0 for all ${\tt SNR}$ and ${\tt INR}$ we must have $\theta = 0$ and  $|\lambda_{24}|=\frac{\sqrt{{\tt SNR}\times {\tt INR}}}{{\tt INR}^2}$ (discussed next).  
Note that we must constrain $|\lambda_{24}|\in [0,1]$. In the interval $|\lambda_{24}| \in \left[0, \frac{\sqrt{\tt SNR\times INR}}{{\tt INR}^2}\right]$ one may verify that the function is increasing in $|\lambda_{24}|$. Thus, if $ \frac{\sqrt{\tt SNR\times INR}}{{\tt INR}^2} \leq 1$, $(|\lambda_{24}| =  \frac{\sqrt{\tt SNR\times INR}}{{\tt INR}^2}, \theta =0)$ maximizes \eqref{lambda}; this happens if ${\tt SNR}\leq {\tt INR}^3$, and yields the first bound in \eqref{R_weak_BE}. Otherwise, for ${\tt SNR}>{\tt INR}^3$, $(\lambda_{24} = 1, \theta=0)$ maximizes \eqref{lambda}, yielding the second equation in \eqref{R_weak_BE}. 
\end{proof}

\begin{remark} 
 The sum-rate bound for $R_{sym\rightarrow}$ of \eqref{R_weak_E} has the same form as Etkin, Tse and Wang's outer bound for one-way Gaussian interference channel \cite[(12)]{etkin_tse_wang} which is useful in weak interference.  The sum-rate bound for $R_{sym\leftarrow}$ is quite different, and we note that it may be verified that \eqref{R_weak_BE} is always at least as large as \eqref{R_weak_E}, as one might expect given the partial adaptation constraints on nodes in the $\rightarrow$ direction, but none on the nodes in the $\leftarrow$ direction.   
\end{remark}

\bigskip

We next show that these outer bounds, derived for the fully adaptive or partially adaptive models, may be achieved to within a constant gap or capacity by {\it non-adaptive} schemes -- i.e. the simultaneous decoding or the Han and Kobayashi scheme operating in the two directions independently. We break our analysis into three sub-sections: 1) very strong interference, 2) strong interference, and 3) weak interference. The overall finite gap results are summarized in Table \ref{table:gaps}.

\subsection{Very Strong Interference: ${\tt INR}\geq {\tt SNR}(1+{\tt SNR})$}
We first show that a non-adaptive scheme may achieve the capacity for the two-way Gaussian IC under a partially adaptive model in very strong interference. For the symmetric two-way Gaussian IC, define ``very strong interference'' as the class of channels for which ${\tt INR}\geq {\tt SNR}(1+{\tt SNR})$, as in \cite[below equation (21)]{etkin_tse_wang}.
%
It is well known that the capacity region of the one-way Gaussian IC in very strong interference is that of two parallel Gaussian point-to-point channels \cite{Carleial1975}, which may be achieved by having each receiver first decode the interfering signal, treating its own as noise, subtracting off the decoded interference, and decoding its own message. Given that the interference is so strong,  this may be done without a rate penalty. We ask whether the same is true for the two-way Gaussian IC with partial adaptation. The answer is affirmative and the capacity region is given by the following theorem:


%

%




\begin{theorem}
\label{2wIC_very_strong}
The capacity region for the two-way Gaussian interference channel with partial adaptation in very strong interference is the set of rate pairs ($R_{12},R_{21},R_{34},R_{43}$), such that \eqref{partial_single_rate1}--\eqref{partial_singe-rate4} are satisfied. 
\end{theorem}
\begin{proof}
Each node may ignore its ability to adapt,  and rather transmit using a ${\cal CN}(0,1)$ Gaussian random code. Each receiver may cancel its own self-interference, and then proceed to first decode the single interfering term before decoding its own message. This standard non-adaptive scheme may achieve the outer bound in \eqref{partial_single_rate1}--\eqref{partial_singe-rate4} in Theorem \ref{thm:outer-partial}.
\end{proof}

Interestingly, the capacity region of the two-way Gaussian interference channel with partial adaptation in very strong interference, is equivalent to the capacity regions of two one-way Gaussian interference channels with very strong interference in parallel and is achieved using a non-adaptive scheme. This allows us to conclude that {\it partial} adaptation is useless in this symmetric and very strong interference regime.

\subsection{Strong Interference: ${\tt SNR}\leq {\tt INR}\leq {\tt SNR}(1+{\tt SNR})$}
In this regime, we are able to show that a non-adaptive scheme may achieve capacity to within a constant gap of any {\it fully} adaptive scheme (in contrast to any {\it partially} adaptive scheme in the last subsection).  A symmetric two-way Gaussian IC,  as in \cite{etkin_tse_wang}, is said to be in ``strong interference'' when ${\tt INR}\geq {\tt SNR}$.

The capacity region of one-way Gaussian interference channel in strong interference is given by \cite{sato_strong}, and for symmetric channels, the capacity region when the interference is strong but not very strong, i.e. ${\tt SNR}\leq {\tt INR}\leq {\tt SNR}(1+{\tt SNR})$,  may be written as 
\begin{align}
R_{sym} = \frac{R_{12}+R_{34}}{2} \leq \frac{1}{2}\log (1+{\tt SNR}+{\tt INR}).\label{R_sato}
\end{align} 
We note that this rate is achievable for the two-way Gaussian IC  by using the simultaneous non-unique decoding scheme for the interference channel in strong interference \cite{ElGamalKim:book, sato_strong, ahlswede1974}) in the $\rightarrow$ and $\leftarrow$ directions, and noting that any self-interference may be canceled. This is a non-adaptive scheme.

We will show that this non-adaptive scheme which achieves \eqref{R_sato} in each direction also achieves to within 1 bit (per user, per direction) of our fully adaptive outer bound \eqref{R_strong_E} in strong but not very strong interference. 

\begin{theorem}
The capacity region for two-way symmetric Gaussian interference channel with full adaptation in strong (but not very strong) interference is within 1 bit to \eqref{R_sato} (per user, per direction)
\end{theorem}
\begin{proof}
\begin{align*}
 \eqref{R_strong_E} - \eqref{R_sato}
&= \frac{1}{2} \log (1+{\tt SNR}+{\tt INR}+2\sqrt{{\tt SNR}\times {\tt INR}})+\frac{1}{2}\log \left(1+\frac{{\tt SNR}}{1+{\tt INR}}\right)-\frac{1}{2}\log \left(1+{\tt SNR}+{\tt INR}\right)\\
&\overset{(a)}{\leq} \frac{1}{2} \log 2(1+{\tt SNR}+{\tt INR})+\frac{1}{2}\log \left(1+\frac{{\tt SNR}}{1+{\tt INR}}\right)-\frac{1}{2}\log (1+{\tt SNR}+{\tt INR})\\
&\overset{(b)}{\leq} \frac{1}{2}+\frac{1}{2}\log \left(1+\frac{{\tt INR}}{{\tt INR}}\right)\\
&=1
\end{align*}
In step (a), we use the fact that $1+{\tt SNR}+{\tt INR}+2\sqrt{{\tt SNR}\times {\tt INR}}\leq 2(1+{\tt SNR}+{\tt INR})$. Step (b) follows from the condition of strong interference ${\tt INR}\geq {\tt SNR}$. Since our bound \eqref{R_strong_E} is valid for the symmetric assumptions of full adaptation, we conclude that the non-adaptive schemes' gap to the fully adaptive outer bound for each user, for each direction is at most 1 bit.
\end{proof}




\begin{remark} Note that if we were to evaluate the fully adaptive outer bound of \eqref{R_strong} under partial adaptation constraints instead, i.e., $X_1$ and $X_3$ are only functions of $M_{12}$ and $M_{34}$ respectively,  then we would be able to set $\lambda_{13}$ in \eqref{R_strong} equal to 0, yielding a new outer bound $R_{sym\rightarrow} \leq \frac{1}{2}\log(1+{\tt SNR}+{\tt INR}) + \frac{1}{2}\log\left(1+\frac{{\tt SNR}}{1+{\tt INR}}\right)$. In this case a gap of $\frac{1}{2}$ bit instead of 1 bit may be shown for  $R_{sym \rightarrow}$. However, due to the asymmetry of partial adaptation ($\lambda_{24}$ in general not equal to 0), in the opposite direction,  we would still have a 1 bit gap for $R_{sym\leftarrow}$. 
\end{remark}

\subsection{Weak Interfererence: ${\tt INR}\leq {\tt SNR}$}
We now show that the well known Han and Kobayashi scheme employed in the $\rightarrow$ and $\leftarrow$ directions may achieve to within a constant number of bits of the fully or partially adaptive (depends on the channel regimes, or relative ${\tt SNR}$ and ${\tt INR}$ values) capacity region for the two-way Gaussian IC. 

\begin{theorem}
 A non-adaptive scheme may achieve to within a $2$ bit per user per direction of partially adaptive capacity region for the two-way Gaussian IC in weak interference. In some channel regimes, this non-adaptive scheme also achieves to within a constant gap of any {\it fully} adaptive scheme.
\end{theorem}
\begin{proof}
As for the one-way IC \cite{etkin_tse_wang}, we break our proof into two regimes: ${\tt INR}\geq 1$ or ${\tt INR}<1$. 

\subsubsection{${\tt INR}\geq 1$}
Outer bounds have already been derived. Consider now using the specific choice of the Han and Kobayashi  (HK) strategy utilized for the symmetric one-way IC as in \cite[(4)]{etkin_tse_wang}  in each direction. That is, view nodes 1,2 as transmitters and 3,4 as receivers in the $\rightarrow$ direction and employ the particular choice of the HK scheme where private messages are encoded at the level of the noise, and similarly for the $\leftarrow$ direction consider nodes 3,4 as transmitters and 1,2 as receivers. Due to the additive nature of the channel and each node's ability to first cancel out their self-interference, one may achieve the following rates per user, per node for each direction when ${\tt INR}\geq 1$ for the symmetric two-way Gaussian IC:
\begin{align}
& R_{HK}=\min \left\{\frac{1}{2}\log (1+{\tt INR}+{\tt SNR})+\frac{1}{2}\log \left(2+\frac{{\tt SNR}}{{\tt INR}}\right)-1, \log \left(1+{\tt INR}+\frac{{\tt SNR}}{{\tt INR}}\right)-1\right\} \label{HK} \\
& =: \min\{R_{HK1}, R_{HK2}\}. \label{HKdef}
\end{align}

\smallskip
\noindent {\bf If the first term in \eqref{HK} is active} we show a constant gap to the outer bound \eqref{R_strong_E},
\begin{align*}
& \eqref{R_strong_E}-R_{HK1}\\
&=\frac{1}{2} \log \left(1+{\tt SNR}+{\tt INR}+2\sqrt{{\tt SNR}\times {\tt INR}}\right)+\frac{1}{2}\log \left(1+\frac{{\tt SNR}}{1+{\tt INR}}\right)\\
&-\frac{1}{2}\log (1+{\tt INR}+{\tt SNR})-\frac{1}{2}\log \left(2+\frac{{\tt SNR}}{{\tt INR}}\right)+1\\
&\leq \frac{1}{2} \log 2(1+{\tt SNR}+{\tt INR})-\frac{1}{2}\log (1+{\tt INR}+{\tt SNR})+\frac{1}{2}\log \left(1+\frac{{\tt SNR}}{{\tt INR}}\right)-\frac{1}{2}\log \left(2+\frac{{\tt SNR}}{{\tt INR}}\right)+1\\
&\leq \frac{1}{2}\log (2)+\frac{1}{2}\log (1)+1\\
&=1.5
\end{align*}

\begin{remark} Since our bound \eqref{R_strong_E} is derived assuming full adaptation,  we may conclude that this gap holds for both $R_{sym\rightarrow}$ and $R_{sym\leftarrow}$ (i.e. holds for $R_{sym}$). 
If  we were to consider partial adaptation ($\lambda_{13}=0$), this gap could be reduced to 1 bit instead of 1.5 bits for $R_{sym\rightarrow}$, but would remain 1.5 bits for $R_{sym\leftarrow}$ as $\lambda_{24}\neq 0$ in general for partial adaptation.\end{remark}

 \smallskip
\noindent {\bf If the second term in \eqref{HK} is active,} {we use outer bound \eqref{R_weak_E} for the forward direction, to bound the gap for $R_{sym\rightarrow}$ as }
\begin{align*}
 \eqref{R_weak_E}-R_{HK2}& =\log \left(1+{\tt INR}+{\tt SNR}-\frac{{\tt INR}\times {\tt SNR}}{1+{\tt INR}}\right)-\log \left(1+{\tt INR}+\frac{{\tt SNR}}{{\tt INR}}\right)+1\\\
& =\log \left(\frac{(1+{\tt INR})^2+{\tt SNR}}{1+{\tt INR}}\frac{{\tt INR}}{{\tt INR}(1+{\tt INR})+{\tt SNR}}\right)+1\\
&= \log \left(\frac{{\tt INR}(1+{\tt INR})^2+{\tt SNR}\times {\tt INR}}{{\tt INR}(1+{\tt INR})^2+{\tt SNR}(1+{\tt INR})}\right)+1\\
&\leq \log (1)+1\\
&=1
\end{align*}
Since our bound \eqref{R_weak_E}  has the same form as the ETW bound \cite{etkin_tse_wang}, the capacity of the Gaussian two-way interference channel with partial adaptation in the forward direction is also to within 1 bit of the specific HK rate \eqref{HK}, \eqref{HKdef} when ${\tt INR}\geq 1$.


{We use outer bound \eqref{R_weak_BE} for the backward direction, to bound the gap for $R_{sym\leftarrow}$,} noting that we need to consider both cases separately. 
{If the first term in \eqref{R_weak_BE} is relevant (${\tt SNR}\leq {\tt INR}^3$):}
\begin{align*}
  \eqref{R_weak_BE} -R_{HK2}
&=\log \left(1+{\tt INR}+\frac{{\tt SNR}}{{\tt INR}}\right)-\log \left(1+{\tt INR}+\frac{{\tt SNR}}{{\tt INR}}\right)+1 =1
\end{align*}

{If the second term in \eqref{R_weak_BE} is relevant (${\tt SNR}\geq {\tt INR}^3$):} 
\begin{align*}
&  \eqref{R_weak_BE} -R_{HK2}\\
&=\log \left(1+\frac{(\sqrt{{\tt SNR}}+\sqrt{{\tt INR}})^2}{1+{\tt INR}}\right)-\log \left(1+{\tt INR}+\frac{{\tt SNR}}{{\tt INR}}\right)+1\\
&=\log \left(\frac{(1+2{\tt INR}+{\tt SNR}+2\sqrt{{\tt SNR}\times {\tt INR}}){\tt INR}}{(1+{\tt INR})((1+{\tt INR}){\tt INR}+{\tt SNR})}\right)+1\\
&\overset{(a)}{\leq }\log \left(\frac{(2(1+{\tt INR}+{\tt SNR})+{\tt INR}){\tt INR}}{(1+{\tt INR})((1+{\tt INR}){\tt INR}+{\tt SNR})}\right)+1\\
&= \log \left(\frac{2{\tt INR}+2{\tt SNR}\times {\tt INR}+3{\tt INR}^2}{{\tt INR}+2{\tt INR}^2+{\tt SNR}+{\tt INR}^3+{\tt SNR}\times {\tt INR}}\right)+1\\
&\leq \log\left(\frac{2({\tt INR}+{\tt SNR}\times {\tt INR}+2{\tt INR}^2+{\tt SNR}+{\tt INR}^3)}{{\tt INR}+{\tt SNR}\times {\tt INR}+2{\tt INR}^2+{\tt SNR}+{\tt INR}^3}\right)+1\\
& = \log (2)+1\\
&=2
\end{align*} 
where (a) follows the fact that $1+{\tt SNR}+{\tt INR}+2\sqrt{{\tt SNR}\times {\tt INR}}\leq 2(1+{\tt SNR}+{\tt INR})$.

\begin{remark} We have shown that the capacity region of the Gaussian two-way interference channel with partial adaptation (fix $X_1$ and $X_3$) is within at most 2 bits per user per direction to the region achieved by two simultaneous HK schemes in opposite directions when ${\tt INR}\geq 1$ for both directions. Again,  we may conclude that partial adaptation cannot significantly increase the capacity for Gaussian two-way IC with weak interference. 
\end{remark}

\subsubsection{${\tt INR}<1$} In this case, a symmetric version of the HK scheme may be obtained 
from  \cite[(69)]{etkin_tse_wang}, for which each of the four users may achieve the following rate:
\begin{align}
& R_{{\tt INR}<1}\leq \log \left(1+\frac{{\tt SNR}}{1+{\tt INR}}\right)\label{R_INR<1}
\end{align}
{\bf We show that this achieves to within 1 bit of the outer bound \eqref{R_strong_E}} 
\begin{align*}
\eqref{R_strong_E}-R_{{\tt INR}<1}
&\leq \frac{1}{2} \log \left(1+{\tt SNR}+{\tt INR}+2\sqrt{{\tt SNR}\times {\tt INR}}\right)+\frac{1}{2}\log \left(1+\frac{{\tt SNR}}{1+{\tt INR}}\right)-\log \left(1+\frac{{\tt SNR}}{1+{\tt INR}}\right)\\
& =\frac{1}{2} \log \left(1+{\tt SNR}+{\tt INR}+2\sqrt{{\tt SNR}\times {\tt INR}}\right)-\frac{1}{2}\log \left(1+\frac{{\tt SNR}}{1+{\tt INR}}\right)\\
& \leq \frac{1}{2}\log \left(\frac{2(1+{\tt SNR}+{\tt INR})(1+{\tt INR})}{1+{\tt SNR}+{\tt INR}}\right)\\
& \overset{(a)}{\leq} \frac{1}{2}\log (4)\\
&=1
\end{align*}
where (a) uses the condition ${\tt INR}<1$. Since \eqref{R_strong_E} was obtained for full adaptation, we conclude that the capacity of the Gaussian two-way IC is to within 1 bit to the HK region when ${\tt INR}<1$ for both directions.
\end{proof}

We summarize our results for constant gaps in Table \ref{table:gaps}.
\begin{table}
\resizebox{18cm}{!} {
\begin{tabular}{|c|c|c|c|c|c|}
\hline
{\bf Two-way Interference} & \multicolumn{5}{r|}{{\bf Constant Gaps per user per direction, in bits (to outer bound)}}\\ \hline
Very Strong & \multicolumn{5}{r|}{0 (partial)}\\ \hline
Strong & \multicolumn{5}{r|}{1 (full)}\\ \hline
 & ${\tt INR}<1$ & \multicolumn{4}{r|}{1 (full)}\\ \cline{2-6}
 &  & HK1 is active & \multicolumn{3}{r|}{1.5 (full)}\\ \cline{3-6}
 Weak & ${\tt INR}\geq 1$ &  & $\rightarrow$ direction & \multicolumn{2}{r|}{1 (partial)}\\ \cline{4-6}
  & & HK2 is active & $\leftarrow$ direction & ${\tt SNR}\leq {\tt INR}^3$ & 1 (partial)\\ \cline{5-6}
  & & & &  ${\tt SNR}> {\tt INR}^3$ & 2 (partial)\\ \hline
\end{tabular}
}
\caption{Constant gaps between non-adaptive symmetric Han and Kobayashi schemes in each direction and partially or fully adaptive outer bounds for the two-way Gaussian IC.}
\label{table:gaps}
\end{table}


\subsection{Final comments on adaptation versus no-adaptation, and versus perfect output feedback}

In the above, we have highlighted classes of two-way interference channels for which adaptation is useless. For the Gaussian channel, only highly symmetric scenarios were considered. The conclusions made for such 
symmetric scenarios, while insightful, do not tell the whole story. That is, it must be noted that adaptation can provide unbounded gains over non-adaptation for certain channels. A simple example of a channel in which we 
see this is shown in Fig. \ref{fig:ex1}. Here the forward channel has no direct links, while the reverse channel has no cross-over links. In this scenario, a non-adaptive scheme would not be able to achieve any positive rate 
for $R_{12}$ and $R_{34}$. However, the adaptive scheme would be able to achieve positive rates for $R_{12}$ and $R_{34}$ by ``routing'' the messages. For example, for message $M_{12}$ to travel from Tx 1 to Rx 2, 
using adaptation it may do so by taking the path: Tx 1 $\rightarrow$  cross-over link to Rx 4 $\rightarrow$  direct reverse link from Tx 4 to Rx 3 $\rightarrow$  cross-over link Tx 3 to Rx 2. An adaptive network may thus 
provide unbounded gain over a non-adaptive network for at least some of the rates. Note however that if the reverse direction is ``routing'' messages for the forward direction, its own message rates will decrease. 

\begin{figure}
\centerline{\includegraphics[width=7cm]{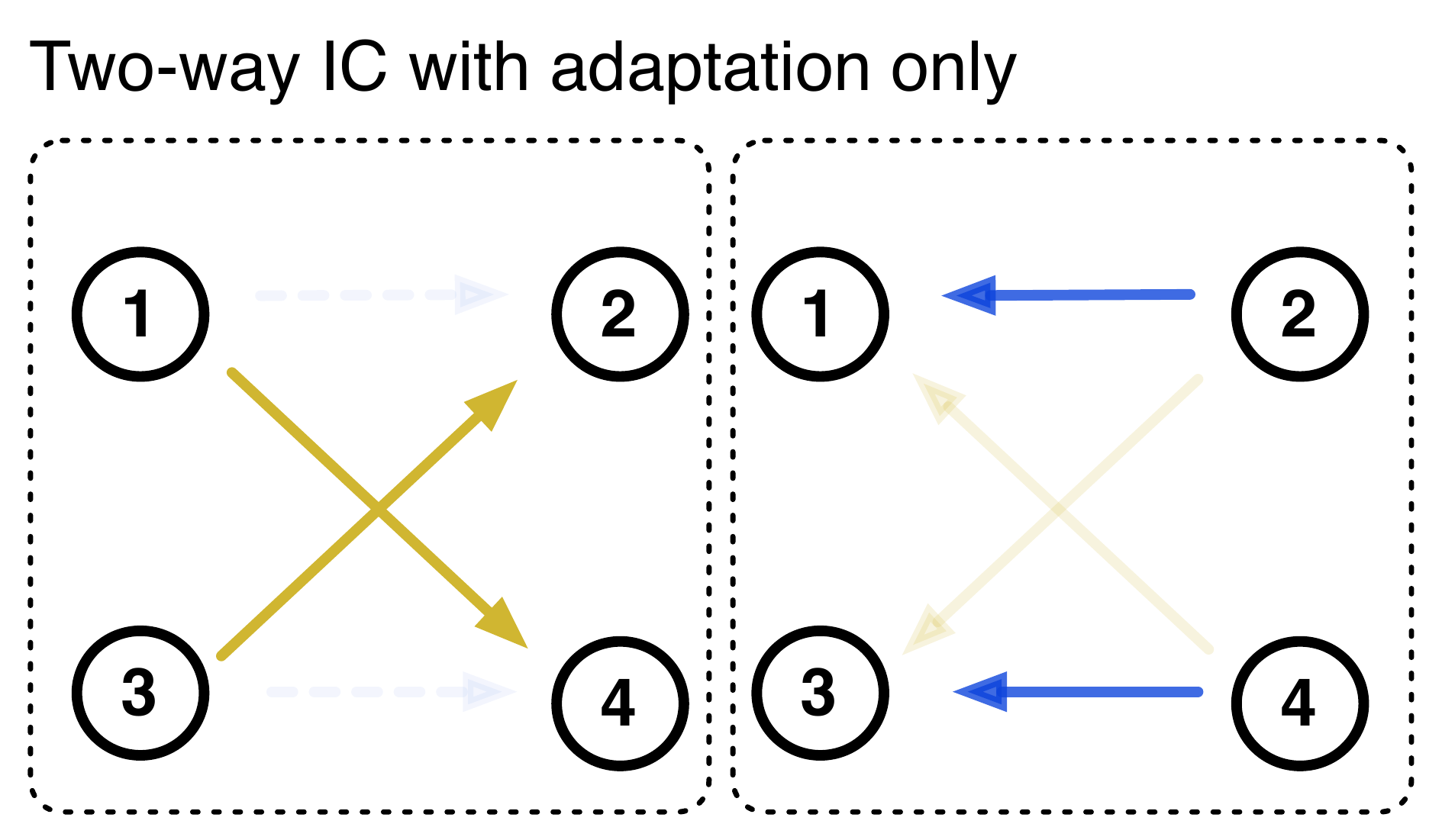}}
\caption{Example of a channel in which adaptation yields unbounded gain over non-adaptation. The non-adaptive scheme would not be able to achieve any non-trivial rates for $R_{12}$ and $R_{34}$, while the adaptive scheme would be able to achieve strictly positive rates (depending on the noise in the different links).}
\label{fig:ex1}
\end{figure}


One can also find examples of networks where perfect output feedback provides unbounded gain over adaptation.
Note that in most scenarios where adaptation is useless, perfect output feedback is known to be useless as well. For example, for the symmetric linear deterministic two-way IC shown in Fig. \ref{fig:curves}, the symmetric sum-rate of the linear deterministic one-way IC and one-way IC {with feedback} are identical for $\frac{2}{3}\leq \alpha \leq 2$, indicating that feedback is useless. In this regime, adaptation was also shown to be useless. In the two-way MAC/BC, for all deterministic models, perfect output feedback may be shown to be useless, and adaptation was also useless.  
One might ask whether feedback and adaptation being useless always go hand in hand. The following example shows the intuitive fact that adaptation being useless does not imply that feedback is useless. 
To show that feedback may provide an unbounded gain over pure adaptation, consider the two-way IC in Fig. \ref{fig:ex2}. For message $M_{12}$ to travel from Tx 1 to Rx 2, using feedback it may do so by taking the path: Tx 1 $\rightarrow$  cross-over link to Rx 4 $\rightarrow$  feedbacks to Tx 3 $\rightarrow$  cross-over link to Rx 2. However, clearly if we employ only adaptation over these forward and reverse links, $M_{12}$ is only able to be decoded by Rx 4 and even with adaptation has no possible way to reach Rx 2. So, feedback may improve the capacity in an unbounded way over adaptation, at least when feedback is ``free'', or perfect (not over other interfering links in the reverse direction). An alternative example of when feedback may outperform adaptation is in the symmetric linear deterministic IC: for $\alpha>2$, as seen in Fig. \ref{fig:curves}, feedback outperforms adaptation. This is intuitive, as feedback is provided over perfect, infinite capacity links, whereas adaptation must take place over the same links over which the data travels. 
Whether  feedback being useless necessarily implies that adaptation is also useless is an interesting open question; all known examples for additive channels seem to suggest this but it has not been rigorously shown.

\begin{figure}
\centerline{\includegraphics[width=14cm]{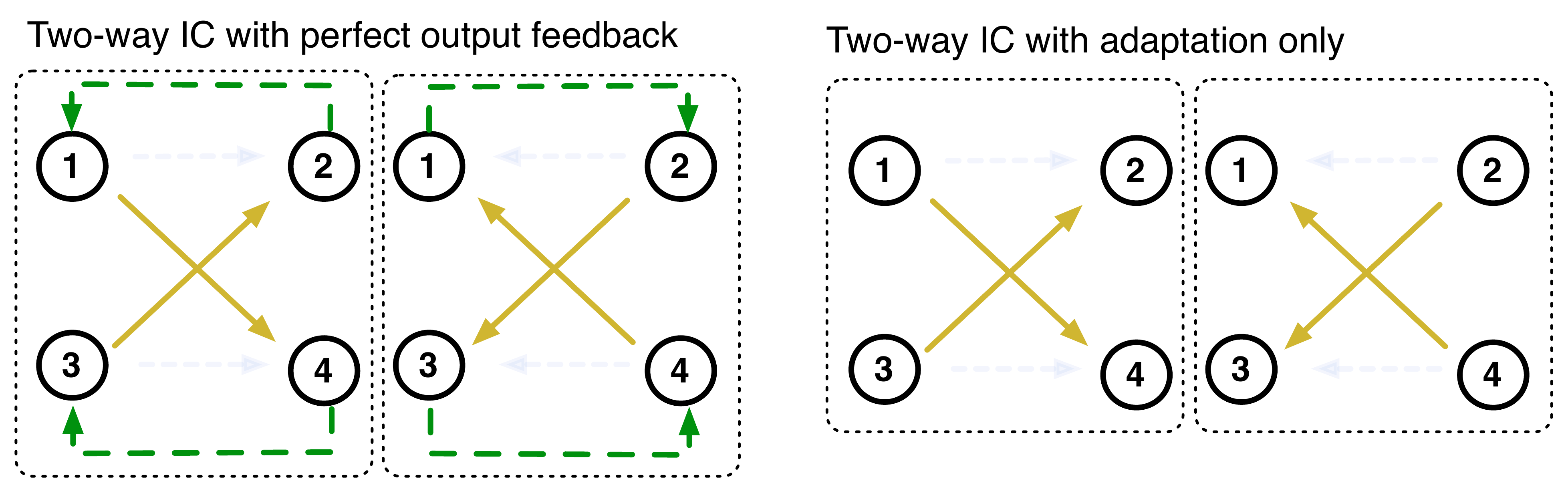}}
\caption{An example of a channel in which perfect output feedback (denoted by the dashed arrows in the left figure) would be able to achieve strictly positive rates for $R_{12}$ and $R_{34}$, but the adaptive scheme for the same channel on the right would not be able to achieve non-zero rates for any of the links.}
\label{fig:ex2}
\end{figure}

\section{Conclusion}
\label{conclusion}
In this work, we have demonstrated a few examples of two-way multi-user channels for which adaptation, or the ability of nodes to adapt their current channel inputs based on previously received channel outputs, is useless from a capacity region perspective, i.e. non-adaptive schemes achieve outer bounds derived for the fully adaptive models. 
Specifically, we obtained the capacity regions of the two-way MAC/BC channel, the two-way Z channel, and the two-way IC of binary modulo-2 addition model, the ``deterministic, invertible and cardinality constrained''  model, and the linear deterministic model. Interestingly, adaptation (full or partial) is not needed to attain the capacity regions even though it is permitted. For noisy channels, we first showed that adaptation can only increase the sum-rate of the two-way Gaussian MAC/BC by up to $\frac{1}{2}$ bit per direction. We then considered the Gaussian two-way IC with 4 terminals and 4 messages. There, it was shown that partial adaptation is useless in very strong interference, and for all other regimes non-adaptive schemes achieved to within constant gaps of fully, or partially, adaptive schemes.
We have demonstrated several examples where adaptation is useless -- the question of when adaptation is useless in general networks remains a challenging open question. However, based on some of the examples seen here, we believe that the following properties may be needed to make the claim that ``adaptation is useless'' for a particular network: 1) the self-interference can be cancelled (excludes the binary multiplier channel), 2)  no loop in the networks (excludes the relaying of data along stronger paths), and 3)  no ``coherent'' gains (excluding possible gains by having users use adaptation to create joint input distributions in for example Gaussian networks). 

\section{Appendix}

\subsection{Proof of Theorem \ref{mod2thm:z}}
\label{mod2Z}

\begin{proof}
Time-sharing may again be used to achieve this region.   For the converse,

\smallskip
\noindent{\it Proof of bound \eqref{z1}:}
\begin{align*}
& n(R_{12}+R_{32}+R_{34}-\epsilon)\\
 & \leq I(M_{12};Y_2^n|M_{21},M_{23},M_{43})+I(M_{32},M_{34};Y_4^n,Y_2^n|M_{43},M_{12},M_{21},M_{23})\\
&\overset{(a)}{\leq} I(M_{12};Y_2^n|M_{21},M_{23},M_{43})+I(M_{32},M_{34};Y_2^n|M_{43},M_{12},M_{21},M_{23})\\
& \ \ +I(M_{32},M_{34};Y_4^n|M_{43},M_{12},M_{21},M_{23},Y_2^n)\\
&\overset{(b)}{\leq} H(Y_2^n|M_{21},M_{23},M_{43})-H(Y_2^n|M_{12},M_{21},M_{23},M_{43})+H(Y_2^n|M_{12},M_{21},M_{23},M_{43})\\
& \ \ +H(Y_4^n|M_{43},M_{12},M_{21},M_{23},Y_2^n)\\
&\overset{(c)}{\leq} \sum_{i=1}^n [H(Y_{2,i})+H(Y_{4,i}|M_{12},M_{21},M_{23},M_{43},Y_4^{i-1},Y_2^n)]\\
& \overset{(d)}{=} \sum_{i=1}^n [H(Y_{2,i})+H(X_{3,i}\oplus X_{4,i}|M_{12},M_{21},M_{23},M_{43},Y_4^{i-1},X_4^i,X_3^{i-1},Y_2^n,X_2^n)]\\
& \overset{(e)}{=} \sum_{i=1}^n [H(Y_{2,i})+H(X_{3,i}|M_{12},M_{21},M_{23},M_{43},Y_4^{i-1}, X_4^i,X_3^{i-1},X_1^n \oplus X_2^n\oplus X_3^n,X_2^n,X_1^n)]\\
&=\sum_{i=1}^n [H(Y_{2,i})]\leq n
\end{align*}

where (a) follows from the chain rule. We drop two negative entropy terms in inequality (b) and notice that the second and the third entropy terms cancel each other. In  (c), we apply the chain rule first, then we drop the conditioning part of the first entropy term. In (d), we construct $X_4^i=f_4(M_{43},Y_4^{i-1})$ and  note that $X_3^{i-1}$ may be obtained from $Y_4^{i-1}=X_3^{i-1}\oplus X_4^{i-1}$, given  $X_4^{i-1}$. Adding $X_2^n$ follows from the fact $X_2^n=f_2(M_{21},M_{23},Y_2^{n-1})$. In (e), we cancel $X_{4,i}$ in the second entropy term since we know $X_4^i$. In addition, given $M_{12}$ and $X_2^n$, we may construct $X_1^n$ as illustrated in Fig. \ref{fig:zchain1}. Now, we may obtain $X_3^n$ from $Y_2^n=X_1^n\oplus X_2^n\oplus X_3^n$, so that the second entropy term in zero. Bound  \eqref{z2} follows by symmetry.

\begin{figure}
\begin{center}
\includegraphics[width=3.6in]{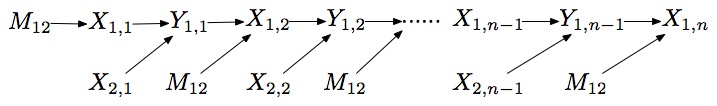}
\caption{The Markov chain used in the outer bound proof of Theorem \ref{mod2thm:z}.} 
\label{fig:zchain1}
\end{center}
\end{figure}

\end{proof}

\subsection{Evaluation of the sum-rate outer bound with full adaptation in Gaussian two-way IC of Theorem \ref{thm:outer-full}}
\label{Gsum}

Letting $E[X_jX_k^*] = \lambda_{jk}$, 
suppressing the subscript $i$, and assuming a symmetric channel, the first two terms  in \eqref{(c)} may be bounded as
\begin{align*}
& H(g_{12}X_1+g_{32}X_3+Z_2|X_2)-H(Z_2)\\
&\leq H(g_{12}X_1+g_{32}X_3+Z_2)-H(Z_2)\\
& \leq \log 2\pi e (\mbox{Var}(g_{12}X_1)+\mbox{Var}(g_{32}X_3)+2\mbox{Cov}(g_{12}X_1,g_{32}X_3)+1)-\log 2\pi e(1)\\
& =\log  ({\tt SNR}+{\tt INR}+2|\lambda_{13}|\cos \theta\sqrt{{\tt SNR}\times {\tt INR}}+1)
\end{align*} 
where $\theta$ is the angle of $g_{12}g_{32}^*\lambda_{13}$.

Similarly, the last two terms may be bounded as
\begin{align}
H(g_{34}X_3+Z_4|g_{32}X_3+Z_2,X_4)-H(Z_4)&\leq \log \left(\frac{\mbox{Var}(g_{34}X_3+Z_4|g_{32}X_3+Z_2,X_4)}{\sigma_4^2}\right)\notag \\
&\leq\log \left(1+\frac{{\tt SNR}(1-|\lambda_{34}|^2)}{{\tt INR}(1-|\lambda_{34}|^2)+1}\right).\label{r12r34}
\end{align}

Combining all terms, we obtain
\begin{align}
R_{sym}=\frac{R_{12}+R_{34}}{2}&\leq \frac{1}{2}\log  ({\tt SNR}+{\tt INR}+2|\lambda_{13}|\cos \theta\sqrt{{\tt SNR}\times {\tt INR}}+1) \notag \\
&+\frac{1}{2}\log \left(1+\frac{{\tt SNR}(1-|\lambda_{34}|^2)}{{\tt INR}(1-|\lambda_{34}|^2)+1}\right) \label{R_strong} 
\end{align}
To obtain \eqref{R_strong_E} one may verify  that \eqref{R_strong} is maximized at $\lambda_{34} = 0$ and $\lambda_{13}=1, \theta=0$.


\bibliographystyle{IEEEtran}
\bibliography{refs}

\end{document}